\documentclass[11pt,letterpaper]{article}

\usepackage{times}
\usepackage{mysetup}
\FULLPAGE

\newcommand{\BobbyNote}[1]{\sidenote{Bobby}{#1}}

\providecommand{\Appendix}{}
\renewcommand{\Appendix}[2][?]{%
	\refstepcounter{section}%
	\vspace{\parskip}%
	{\flushright\Large\bfseries\appendixname\ \thesection: #1}%
	\vspace{\baselineskip}%
}

\renewcommand{\appendix}{%
	\renewcommand{\section}{\secdef\Appendix\Appendix}%
	\renewcommand{\thesection}{\Alph{section}}%
	\setcounter{section}{0}%
}

\title{{\bf Multi-Armed Bandits in Metric Spaces}\footnote{The conference version~\cite{LipschitzMAB-stoc08} of this paper has appeared in ACM STOC 2008. This is the full version.}}

\author{
{\Large Robert Kleinberg}\thanks{Computer Science Department, Cornell University, Ithaca, NY 14853. Email:~{\tt rdk at cs.cornell.edu}. Supported in part by NSF awards CCF-0643934 and  CCF-0729102.}
\and
{\Large Aleksandrs Slivkins}\thanks{
Microsoft Research, Mountain View, CA 94043.
Email:~{\tt slivkins at microsoft.com}.
Parts of this work were done while the author was a postdoctoral research associate at Brown University.}
\and
{\Large Eli Upfal}\thanks{
Computer Science Department, Brown University, Providence, RI 02912. Email:~{\tt eli at cs.brown.edu}.
Supported in part by NSF awards CCR-0121154 and DMI-0600384, and ONR Award N000140610607.
}}

\date{November 2007\\Revised: April 2008, September 2008}

\begin{document}

\newcommand{\B}{\ensuremath{\mathcal{B}}}

\newcommand{\standardMAB}{standard Lipschitz MAB problem}
\newcommand{\quasimetric}{quasi-distance}
\newcommand{\noisyNN}{nearest-neighbor MAB problem}

\newcommand{\DBL}{\mathtt{DBL}} 	
\newcommand{\COV}{\mathtt{COV}}		
\newcommand{\MinDBL}{\mathtt{MinDBL}}	
\newcommand{\MinCOV}{\mathtt{MinCOV}}	
\newcommand{\RegretDim}{\mathtt{DIM}}	
\newcommand{\MaxMinCOV}{\mathtt{MaxMinCOV}}	
\maketitle

\begin{abstract}
In a multi-armed bandit problem, an online algorithm
chooses from a set of strategies in a sequence of
$n$ trials so as to maximize the total payoff of the
chosen strategies.
While the performance of bandit algorithms with
a small finite strategy set is quite well understood,
bandit problems with large strategy sets are still a
topic of very active
investigation, motivated by practical applications such as
online auctions and web advertisement.
The goal of such research is to identify broad and natural
classes of strategy sets and payoff functions which enable
the design of efficient solutions.

In this work we study a very general setting for the
multi-armed bandit problem in which the strategies form
a metric space, and the payoff function satisfies a Lipschitz
condition with respect to the metric. We refer to this
problem as the {\it Lipschitz MAB problem}. We present a
solution for the multi-armed problem in this setting.
That is, for every metric space $(L,X)$
we define an isometry invariant $\MaxMinCOV(X)$ which
bounds from below the performance of Lipschitz
MAB algorithms for $X$, and we present an algorithm which
comes arbitrarily close to meeting this bound. Furthermore, our
technique gives even better results for benign payoff functions.
\end{abstract}

\section{Introduction}
\label{sec:intro}

\newcommand{\willcite}[1][Cite]{{[\sc #1]}}

In a multi-armed bandit problem, an online algorithm
must choose from a set of strategies in a sequence of
$n$ trials so as to maximize the total payoff of the
chosen strategies.
These problems are the principal theoretical
tool for modeling the exploration/exploitation tradeoffs
inherent in  sequential
decision-making under uncertainty.  Studied
intensively for the last three
decades~\cite{Berry-book,CesaBL-book,Gittins-index},
bandit problems are having an increasingly visible
impact on computer science because of their diverse
applications including online auctions, adaptive
routing, and the theory of learning in games.
The performance of a multi-armed bandit algorithm is often evaluated in terms of its {\it regret}, defined as the gap between the expected payoff of the algorithm and that of an optimal strategy.
While the performance of bandit algorithms with
a small finite strategy set is quite well understood,
bandit problems with exponentially or infinitely
large strategy sets are still a topic of very active
investigation~\cite{Agrawal-bandits-95,bandits-exp3,AuerOS/07,Bobby-stoc04-journal-version,sundaram-bandits-92,Cope/06,DaniHK/07,DaniH/06,FlaxmanKM/05,KakadeKL/07,Bobby-nips04,Bobby-thesis,McMahanB/04}.

Absent any assumptions about the
strategies and their payoffs, bandit problems
with large strategy sets
allow for
no
non-trivial solutions --- any
multi-armed bandit algorithm
performs as badly, on some inputs, as random guessing.
But in most applications it is natural to assume
a structured class of payoff functions, which
often enables the design of efficient learning
algorithms~\cite{Bobby-thesis}.  In this
paper, we consider a broad and natural class of problems
in which the structure is induced by a metric
on the space of strategies.  While bandit problems have been studied
in a few specific metric spaces (such as a one-dimensional
interval)
~\cite{Agrawal-bandits-95,AuerOS/07,Cope/06,Bobby-nips04,yahoo-bandits07},
the case of general metric spaces has not been treated before,
despite being an extremely natural setting for bandit
problems.  As a motivating example,
consider the problem faced by a website choosing from
a database of thousands of banner ads to display to
users, with the aim of maximizing the
click-through rate of the ads displayed by matching ads
to users' characterizations and the web content that they
are currently watching.
Independently experimenting with each
advertisement is infeasible, or at least highly
inefficient, since the number of ads is too large.
Instead, the advertisements are
usually organized into a taxonomy based on metadata
(such as the category of product being advertised) which
allows a similarity measure to be defined.  The
website can then attempt to optimize its learning
algorithm by generalizing from experiments with
one ad to make inferences about the performance
of similar ads~\cite{yahoo-bandits07, yahoo-bandits-icml07}.
Abstractly, we have a bandit
problem of the following form: there is a strategy
set $X$, with an unknown payoff function
$\mu \,:\, X \rightarrow [0,1]$ satisfying a set
of predefined constraints of the form
$|\mu(u) - \mu(v)| \leq \delta(u,v)$ for some
$u,v \in X$ and $\delta(u,v) > 0$.  In each
period the algorithm chooses a point
$x \in X$ and observes an independent random sample from a
payoff distribution whose expectation is  $\mu(x)$.

A moment's thought reveals that this abstract problem
can be regarded as a bandit problem in a metric space.
Specifically, if $L(u,v)$ is defined to be the
infimum, over all finite sequences
$u = x_0, x_1, \ldots, x_k = v$ in $X$,
of the quantity $\sum_i \delta(x_i,x_{i+1})$,
then $L$ is a metric\footnote{More
precisely, it is a pseudometric  because
some pairs of distinct points $x,y \in X$ may satisfy
$L(x,y)=0$.} and the constraints $|\mu(u)-\mu(v)|
< \delta(u,v)$ may be summarized by stating that
$\mu$ is a Lipschitz function (of Lipschitz constant
$1$) on the metric space $(L,X)$.
We refer to this problem as the
\emph{Lipschitz MAB problem} on $(L,X)$, and
we refer to the ordered triple $(L,X,\mu)$ as
an \emph{instance} of the Lipschitz MAB problem.\footnote{
When the metric space $(L,X)$ is understood from
context, we may also refer to $\mu$ as an instance.}

\xhdr{Prior work.}
\label{sec:intro.priorwork}
While our work is the first to treat the Lipschitz MAB
problem in general metric spaces, special cases of the
problem are implicit in prior work on the continuum-armed
bandit problem~\cite{Agrawal-bandits-95,AuerOS/07,Cope/06,Bobby-nips04}
--- which corresponds to the space $[0,1]$ under the metric
$L_d(x,y) = |x-y|^{1/d}$, $d \geq 1$ --- and the experimental work
on ``bandits for taxonomies''~\cite{yahoo-bandits07},
which corresponds to the
case in which $(L,X)$ is a tree metric.
Before describing our results in greater detail, it is
helpful to put them
in context by recounting the nearly optimal bounds for the
one-dimensional continuum-armed bandit problem, a problem
first formulated by R.~Agrawal in 1995~\cite{Agrawal-bandits-95}
and recently solved (up to logarithmic factors)
by various authors~\cite{AuerOS/07,Cope/06,Bobby-nips04}.
In the following theorem
and throughout this paper, the \emph{regret} of a multi-armed
bandit algorithm $\A$ running on an instance $(L,X,\mu)$ is
defined to be the function $R_{\A}(t)$ which measures
the difference between its expected payoff at time $t$
and the quantity $t\,\sup_{x\in X} \mu(x)$.
The latter quantity is the expected payoff of always playing
a strategy $x \in \argmax \mu(x)$
if such strategy exists.

\begin{theorem}[\cite{AuerOS/07,Cope/06,Bobby-nips04}] \label{thm:intro.1d}
For any $d \geq 1$, consider the Lipschitz MAB
problem on $(L_d,[0,1])$. There is an algorithm $\mathcal{A}$ whose regret on any instance $\mu$ satisfies
	$ R_\A(t) = \tilde{O}(t^\gamma)$
for every $t$, where $\gamma = \frac{d+1}{d+2}$.
No such algorithm exists for any $\gamma < \frac{d+1}{d+2}$.
\end{theorem}

\OMIT{ 
For any $d \geq 1$,
there is an algorithm $\mathcal{A}$ for the Lipschitz MAB
problem on $(L_d,[0,1])$ whose regret on any instance
$\mu$ satisfies
\[
R_{\A}(t) = O \left( t^{\frac{d+1}{d+2}} \log^{\frac{1}{d+2}}(t) \right).
\]
For any $\gamma < \frac{d+1}{d+2}$ there does not exist
an algorithm $\mathcal{A}$ for the Lipschitz MAB problem
on $(L_d,[0,1])$ which satisfies $R_{\A}(t) = O(t^{\gamma})$
for every $t$ and every instance $\mu.$
} 

In fact, if the time horizon $t$ is known in advance,
the upper bound in the theorem can be achieved
by an extremely na\"{i}ve algorithm which simply
uses an optimal $k$-armed bandit algorithm (such as the
\textsc{ucb1} algorithm~\cite{bandits-ucb1}) to choose
strategies from the set
	$S = \{0, \tfrac1k,\, \tfrac2k,\,\ldots,1\}$,
for a suitable choice of the parameter $k$.
While the regret bound in Theorem~\ref{thm:intro.1d} is
essentially optimal for the Lipschitz MAB problem
in $(L_d,[0,1])$, it is strikingly odd that it is
achieved by such a simple algorithm.  In particular,
the algorithm approximates the strategy set by a
fixed mesh $S$ and does not refine
this mesh as it gains information about the location
of the optimal strategy.  Moreover, the metric
contains seemingly useful proximity information,
but the algorithm ignores this information after
choosing its initial mesh.  Is this really the best
algorithm?

A closer examination of the lower bound proof
raises further reasons for suspicion: it is
based on a contrived, highly singular payoff function
$\mu$ that alternates between being constant on some
distance scales and being very steep on other (much smaller)
distance scales, to create a multi-scale ``needle in haystack''
phenomenon which nearly obliterates the usefulness
of the proximity information contained in the
metric $L_d$.  Can we expect algorithms to do better
when the payoff function is more benign?  For the
Lipschitz MAB problem on $(L_1,[0,1])$, the question
was answered affirmatively in~\cite{Cope/06, AuerOS/07}
for some classes of instances,
with algorithms that are tuned to the specific classes.

\OMIT{ 
For the
Lipschitz MAB problem on $(L_1,[0,1])$, the question
was answered affirmatively by Cope~\cite{Cope/06} and
an even stronger affirmative answer was provided by
Auer \emph{et. al.}~\cite{AuerOS/07}.  For example,
a special case of the main result in~\cite{AuerOS/07}
shows that if the payoff function $\mu$ is twice
differentiable with finitely many maxima each having
a nonzero second derivative, then regret $R_{\A}(t) =
O \left( \sqrt{t \log t} \right)$ can be achieved by
modifying the na\"{i}ve algorithm described above
to sample uniformly at random from the interval
$\left[ j/k, (j+1)/k \right)$ instead
of deterministically playing $j/k$.
Our Theorem~\ref{thm:intro.zooming}, stated below,
reveals a similar phenomenon in general metric
spaces: it is possible to define algorithms whose
regret outperforms the per-metric optimal algorithm
when the input instance is sufficiently benign.
} 

\xhdr{Our results and techniques.}
\label{sec:intro.results}
In this paper we consider the Lipschitz MAB problem on arbitrary metric spaces.
We are concerned with the following two main questions motivated by the discussion above:
\begin{itemize}
\item[(i)] What is the best possible bound on regret for a given metric space?

\item[(ii)] Can one take advantage of benign payoff functions?
\end{itemize}
In this paper we give a complete solution to (i), by describing for every metric space $X$ a family of algorithms which come arbitrarily close to achieving the best possible regret bound for $X$. We also give a satisfactory answer to (ii); our solution is arbitrarily close to optimal in terms of the zooming dimension defined below.
In fact, our algorithm for (i) is an extension of the algorithmic technique used to solve (ii).

\OMIT{ 
Our main technical contribution is a new algorithm,
the \emph{zooming algorithm},
that combines the upper confidence bound
technique used in earlier bandit algorithms
such as \textsc{ucb1} with a novel \emph{adaptive
refinement} step that uses past history to zoom
in on regions near the apparent maxima of $\mu$
and to explore a denser mesh of strategies in
these regions.  This algorithm is a key ingredient
in our design of an optimal bandit algorithm for
every metric space $(L,X)$.  Moreover,
we show that the zooming algorithm can perform significantly
better on benign problem instances.  That is, for every instance
$(L,X,\mu)$ we define a parameter called the \emph{zooming
dimension} which is often significantly smaller than
$\MaxMinCOV(X)$, and we bound the algorithm's performance
in terms of the zooming dimension of the problem instance.
Since the zooming algorithm is self-tuning, it achieves
this bound without requiring prior knowledge of the
zooming dimension.
} 

Our main technical contribution is a new algorithm,
the \emph{zooming algorithm},
that combines the upper confidence bound
technique used in earlier bandit algorithms
such as \textsc{ucb1} with a novel \emph{adaptive
refinement} step that uses past history to zoom
in on regions near the apparent maxima of $\mu$
and to explore a denser mesh of strategies in
these regions.  This algorithm is a key ingredient
in our design of an optimal bandit algorithm for
every metric space $(L,X)$.  Moreover,
we show that the zooming algorithm can perform significantly
better on benign problem instances.  That is, for every instance
$(L,X,\mu)$ we define a parameter called the \emph{zooming
dimension}, and use it to bound the algorithm's performance
in a way that is often significantly stronger than the corresponding per-metric bound. Note that the zooming algorithm is \emph{self-tuning}, i.e. it achieves this bound without requiring prior knowledge of the zooming dimension.

To state our theorem on the per-metric optimal solution for (i),
we need to sketch a few definitions which arise naturally as one
tries to extend the lower bound from~\cite{Bobby-nips04} to general
metric spaces.  Let us say that a subset $Y$ in a metric space $X$
has covering dimension $d$ if it can
be covered by $O(\delta^{-d})$ sets of diameter $\delta$
for all $\delta>0$. A point $x\in X$ has local covering
dimension $d$ if it has an open neighborhood of covering dimension $d$.
The space $X$ has max-min-covering
dimension $d = \MaxMinCOV(X)$
if it has no subspace whose local covering dimension
is uniformly bounded below by a number greater than $d$.

\begin{theorem} \label{thm:intro.pmo}
Consider the Lipschitz MAB problem on a compact metric space $(L,X)$. Let
    $d = \MaxMinCOV(X)$.
If $\gamma > \tfrac{d+1}{d+2}$ then there exists a bandit algorithm $\A$ such that for every problem instance $\mathcal{I}$ it satisfies
	$R_{\A}(t) = O_{\mathcal{I}}(t^\gamma)$ for all $t$.
No such algorithm exists if $d>0$ and $\gamma < \tfrac{d+1}{d+2}$.
\end{theorem}

\OMIT{ 
For metric spaces which are highly homogeneous (in the sense
that any two \eps-balls are isometric to one another) the theorem follows easily from a refinement of the techniques introduced in~\cite{Bobby-nips04};
in particular, the upper bound can be achieved using a
generalization of the na\"{i}ve algorithm described earlier.
} 

In general $\MaxMinCOV(X)$ is bounded above by the covering dimension of $X$. For metric spaces which are highly homogeneous (in the sense
that any two \eps-balls are isometric to one another) the two dimensions are equal, and the upper bound in the theorem can be achieved using a  generalization of the na\"{i}ve algorithm described earlier.
The difficulty in Theorem~\ref{thm:intro.pmo} lies in dealing with inhomogeneities in the metric space.\footnote{To appreciate this issue, it is very instructive to consider a concrete example of a metric space $(L,X)$ where
	$\MaxMinCOV(X)$
is strictly less than the covering dimension, and for this specific example design a bandit algorithm whose regret bounds are better than those suggested by the covering dimension. This is further discussed in Section~\ref{sec:pmo}.}
It is important to treat the problem
at this level of generality, because some of the
most natural applications of the Lipschitz MAB problem, e.g.
the web advertising problem described earlier, are
based on highly inhomogeneous metric spaces.  (That is, in web taxonomies,
it is unreasonable to expect different categories at the same
level of a topic hierarchy to have the roughly the same number
of descendants.)

The algorithm in Theorem~\ref{thm:intro.pmo} combines the zooming
algorithm described earlier with a delicate transfinite construction
over closed subsets consisting of ``fat points'' whose local covering
dimension exceeds a given threshold $d$.  For the lower bound, we
craft a new dimensionality notion, the max-min-covering dimension
introduced above,
which captures the inhomogeneity of a metric space, and we connect this
notion with the transfinite construction that underlies the algorithm.

For ``benign'' input instances we provide a  better
performance guarantee for the zooming algorithm.
The lower bounds in Theorems~\ref{thm:intro.1d}
and~\ref{thm:intro.pmo} are based on contrived,
highly singular, ``needle in haystack'' instances
in which the set of near-optimal strategies is
astronomically larger than the set of precisely
optimal strategies. Accordingly, we quantify the tractability
of a problem instance in terms of the number of near-optimal strategies.
We define the \emph{zooming
dimension} of an instance $(L,X,\mu)$
as the smallest $d$ such that the following covering property holds:
for every $\delta>0$ we require only $O(\delta^{-d})$
sets of diameter $\delta/8$ to cover
the set of strategies whose payoff falls short of
the maximum by an amount between $\delta$ and $2 \delta$.

\begin{theorem} \label{thm:intro.zooming}
If $d$ is the zooming dimension of a Lipschitz MAB
instance then at any time $t$ the zooming
algorithm suffers regret
	$\tilde{O}(t^{\gamma})$, $\gamma = \tfrac{d+1}{d+2}$.
Moreover, this is the best possible exponent $\gamma$ as a function of $d$.
\end{theorem}

The zooming dimension can be significantly smaller than the max-min-covering dimension.\OMIT{\footnote{One can show that in this case the na\"{i}ve algorithm from Theorem~\ref{thm:intro.1d} performs poorly compared to the zooming algorithm.}}
Let us illustrate this point with two examples (where for simplicity the max-min-covering dimension is equal to the covering dimension).
\OMIT{ 
First, if $(L,X)$ is the Euclidean metric on a unit interval,
and $\mu$ is a twice-differentiable function with negative
second derivative at the optimal strategy $x^*$, then the zooming
dimension is only $\tfrac12$ whereas the covering dimension is $1$.
} 
For the first example, consider a metric space consisting of a high-dimensional part and a low-dimensional part. For concreteness, consider a rooted tree $T$ with two top-level branches $T'$ and $T''$ which are complete infinite $k$-ary trees, $k=2,10$. Assign edge weights in $T$ that are exponentially decreasing with distance to the root, and let $L$ be the resulting shortest-path metric on the leaf set $X$.\footnote{Here a \emph{leaf} is defined as an  infinite path away from the root.} If there is a unique optimal strategy that lies in the low-dimensional part $T'$ then the zooming dimension is bounded above by the covering dimension of $T'$, whereas the ``global'' covering dimension is that of $T''$.
In the second example, let $(L,X)$ be a homogeneous high-dimensional metric, e.g. the Euclidean metric on the unit $k$-cube, and the payoff function is
	$\mu(x) = 1-L(x,S)$
for some subset $S$. Then the zooming dimension is equal to the covering dimension of $S$, e.g. it is $0$ if $S$ is a finite point set.

\OMIT { 
...  $X=[0,1]$ with the standard metric $L(x,y) = |x-y|$,
... than the local covering dimension at the point $x^*$
where $\mu$ is maximized.
} 


\xhdr{Discussion.}
In stating the theorems above, we have been imprecise about specifying the model of computation.  In particular, we have ignored the thorny issue of how to provide an algorithm with an input containing a metric space which may have an infinite number of points. The simplest way to interpret our theorems is to ignore implementation details and interpret ``algorithm'' to mean an abstract decision rule, i.e. a (possibly randomized) function mapping a history of past observations $(x_i,r_i) \in X \times [0,1]$ to a strategy $x \in X$ which is played in the current period. All of our theorems are valid under this interpretation, but they can also be made into precise algorithmic results provided that the algorithm is given appropriate oracle access to the metric space.  In most cases, our algorithms require only a \emph{covering oracle} which takes a finite collection of open balls and either declares that they cover $X$ or outputs an uncovered point. We refer to this setting as the \standardMAB.  For example, the zooming algorithm uses only a covering oracle for $(L,X)$, and requires only one oracle query per round (with at most $t$ balls in round $t$). However, the per-metric optimal algorithm in Theorem~\ref{thm:intro.pmo}  uses more complicated oracles, and we defer the definition of these oracles to Section~\ref{sec:pmo}.

\OMIT{ 
the algorithm is very efficient, requiring only $O(t \log t)$
operations in total (including oracle queries) to choose
its first $t$ strategies.
\BobbyNote{Prove the $O(t \log t)$ bound in "body".}
} 

While our definitions and results so far have been tailored for the Lipschitz MAB problem on infinite metrics, some of them can be extended to the finite case as well. In particular, for the zooming algorithm we obtain sharp results (that are meaningful for both finite and infinite metrics) using a more precise, \emph{non-asymptotic} version of the zooming dimension. Extending the notions in Theorem~\ref{thm:intro.pmo} to the finite case is an open question.

\OMIT{
While our definitions and results so far have been tailored for the Lipschitz MAB problem on infinite metrics, they can be extended to the finite case as well. In particular, for the zooming algorithm we obtain sharp results (that are meaningful for both finite and infinite metrics) using a more precise, \emph{non-asymptotic} version of the zooming dimension. Extending the notions in Theorem~\ref{thm:intro.pmo} to the finite case is feasible but more complicated; we leave it to the full version.
} 


\xhdr{Extensions.}
We provide a number of extensions in which we elaborate on our analysis of the zooming algorithm.
First, we provide sharper bounds for several examples in which the reward from playing each strategy $u$ is $\mu(u)$ plus an independent \emph{noise} of a known and ``benign" shape. Second, we upgrade the zooming algorithm so that it satisfies the guarantee in Theorem~\ref{thm:intro.zooming} {\em and} enjoys a better guarantee if the maximal reward is exactly 1. Third, we apply this result to a version where
	$\mu(\cdot) = 1-L(\,\cdot\,,S)$
for some \emph{target set} $S$ which is not revealed to the algorithm. Fourth, we relax some assumptions in the analysis of the zooming algorithm, and use this generalization to analyze the version in which
	$\mu(\cdot) = 1-f(L(\,\cdot\,,S))$
for some known function $f$. Finally, we extend our analysis from reward distributions supported on $[0,1]$ to those with unbounded support and finite absolute third moment.

\OMIT{
Some of our initial motivation for this project came from the online advertizing scenario described in the introduction. We follow this motivation further in Appendix~\ref{sec:admatching} and consider a multi-round game such that in each round an adversary selects a webpage and the algorithm selects an ad which it places on this webpage. We assume that we have a Lipschitz condition on the product (webpages$\times$ads) space, and we give an algorithm whose regret dimension (as defined in Section~\ref{sec:pmo}) is upper-bounded in terms of (essentially) the covering dimension. Although the algorithm is based in the ``na\"{i}ve'' algorithm from Theorem~\ref{thm:intro.1d}, the adversarial aspect of the problem creates considerable technical challenges.  In future work we hope to pursue more refined guarantees in the style of Section~\ref{sec:adaptive-exploration}.
}

\OMIT{ 
Ideally, it would be desirable to have a matching
lower bound constituting a \emph{per-instance optimality}
guarantee for the zooming algorithm or some other
algorithm.  The goal, when stated in this form, is
plainly unachievable.  For any given instance
$(L,X,\mu)$, if $x^* \in X$ is a point where $\mu$
achieves its maximum, then the algorithm
which always plays strategy $x^*$ has zero
regret.  Nevertheless, one might hope for a subtler
characterization of per-instance optimality, e.g.
asserting that no algorithm can outperform $\A$ on
one instance $(L,X,\mu)$ without performing significantly
worse than $\A$ on highly similar instances $(L,X,\mu')$.
While we have been unable to prove such guarantees
for the zooming algorithm, the question of per-instance
optimality is an attractive topic for further investigation.
} 

\xhdr{Follow-up work.} For metric spaces whose max-min-covering dimension is exactly 0, this paper provides an upper bound
    $R(T) = O_\mathcal{I}(T^\gamma)$
for any $\gamma>\tfrac12$, but no matching lower bound. Characterizing the optimal regret for such metric spaces remained an open question. Following the publication of the conference version, this question has been settled in~\cite{LipschitzMABdichotomy-TR}, revealing the following dichotomy: for every metric space, the optimal regret of a Lipschitz MAB algorithm is either bounded above by any $f\in \omega(\log t)$, or bounded below by any $g\in o(\sqrt{T})$, depending on whether the completion of the metric space is compact and countable.

\subsection{Preliminaries}
Given a metric space, $B(x,r)$ denotes an open ball of radius $r$ around point $x$. Throughout the paper, he constants in the $O(\cdot)$ notation are absolute unless specified otherwise.

\begin{definition}\label{def:lipschitzMAB}
In the \emph{Lipschitz MAB problem} on $(L,X)$, there is a strategy set $X$, a metric space $(L,X)$ of diameter $\leq 1$, and a payoff function
	$\mu \,:\, X \rightarrow [0,1]$
such that the following \emph{Lipschitz condition} holds:
\begin{align}\label{eq:Lipschitz-condition}
	|\mu(x)-\mu(y)| \leq L(x,y)	\quad
		\text{for all $x,y\in X$}.
\end{align}

\noindent Call $L$ is the \emph{similarity function}. The metric space $(L,X)$ is revealed to an algorithm, whereas the payoff function $\mu$ is not. In each round the algorithm chooses a strategy $x \in X$ and observes an independent random sample from a
payoff distribution $\mathcal{D}(x)$ with support
	$\mathcal{S} \subset [0,1]$ and expectation $\mu(x)$.

The \emph{regret} of a bandit algorithm $\A$ running on a given problem instance is
	$R_{\A}(t) = W_\A(t) - t\mu^*$,
where $W_\A(t)$ is the expected payoff of \A\ at time $t$
and
	$\mu^* = \sup_{x\in X} \mu(x)$
is the \emph{maximal expected reward}.

The \emph{$C$-zooming dimension} of the problem instance $(L,X,\mu)$ is the smallest $d$ such that for every $r \in(0,1]$ the set
	$X_r = \{x\in X:\, \tfrac{r}{2} < \mu^* - \mu(x) \leq r  \}$
can be covered by $C\,r^{-d}$ sets of diameter at most $r/8$.
\end{definition}

\begin{definition}
Fix a metric space on set $X$. Let $N(r)$ be the smallest number of sets of diameter $r$ required to cover $X$. The \emph{covering dimension} of $X$ is
$$ \COV(X) =  \inf \{\,d:\; \exists c\; \forall r>0\quad
	N(r) \leq c r^{-d} \,\}.
$$
The \emph{$c$-covering dimension} of $X$ is defined as the infimum of all $d$ such that
	$N(r) \leq c r^{-d}$
for all $r>0$.
\end{definition}

\xhdr{Outline of the paper.} In Section~\ref{sec:adaptive-exploration} we prove Theorem~\ref{thm:intro.zooming}. In Section~\ref{sec:pmo} we discuss the per-metric optimality and prove Theorem~\ref{thm:intro.pmo}. Section~\ref{sec:gen-confRad} covers the extensions.


\newcommand{\Czoom}{C}
\newcommand{\Phase}{\ensuremath{\mathtt{ph}}}

\section{Adaptive exploration: the zooming algorithm}
\label{sec:adaptive-exploration}

In this section we introduce the {\it zooming algorithm} which uses adaptive exploration to take advantage of the ''benign" input instances, and prove the main guarantee (Theorem~\ref{thm:intro.zooming}).

Consider the \standardMAB\ on $(L,X)$. The zooming algorithm proceeds in phases $i=1,2,3,\ldots$ of $2^i$ rounds each. Let us consider a single phase $i_\Phase$ of the algorithm. For each strategy $v\in X$ and time $t$, let $n_t(v)$ be the number of times this strategy has been played in this phase before time $t$, and let $\mu_t(v)$ be the corresponding average reward. Define $\mu_t(v)=0$ if $n_t(v)=0$.  Note that at time $t$ both quantities are known to the algorithm. Define the \emph{confidence radius} of $v$ at time $t$ as
\begin{equation}\label{eq:confidence-radius}
	r_t(v) := \sqrt{8\, i_\Phase\,/\,(2+n_t(v))}.
\end{equation}

Let $\mu(v)$ be the expected reward of strategy $v$. Note that
	$E[\mu_t(v)] = \mu(v)$.
Using Chernoff Bounds, we can bound $|\mu_t(v) - \mu(v)|$ in terms of the confidence radius:

\begin{definition}\label{def:clean-phase}
A phase is called \emph{clean} if for each strategy $v\in X$ that has been played at least once during this phase and each time $t$ we have
	$ |\mu_t(v) - \mu(v)| \leq r_t(v)$.
\end{definition}

\begin{claim}\label{cl:conf-rad}
Phase $i_\Phase$ is clean with probability at least $1-4^{-i_\Phase}$.
\end{claim}

Throughout the execution of the algorithm, a finite number of strategies are designated \emph{active}. Our algorithm only plays active strategies, among which it chooses a strategy $v$ with the maximal \emph{index}
\begin{equation}\label{eq:index}
	I_t (v) = \mu_t(v) + 2\, r_t(v).
\end{equation}
Say that strategy $v$ \emph{covers} strategy $u$ at time $t$ if
	$u\in B(v,\, r_t(v))$.
Say that a strategy $u$ is \emph{covered} at time $t$ if at this time it is covered by some active strategy $v$. Note that the \emph{covering oracle} (as defined in Section~\ref{sec:intro}) can return a strategy which is not covered if such strategy exists, or else inform the algorithm that all strategies are covered.
Now we are ready to state the algorithm:

\begin{algorithm}[Zooming Algorithm]\label{alg:nice-1d}
Each phase $i$ runs for $2^i$ rounds. In the beginning of the phase no strategies are active. In each round do the following:
\begin{OneLiners}
\item[1.] If some strategy is not covered, make it active.
\item[2.] Play an active strategy with the maximal index~\refeq{eq:index}; break ties arbitrarily.
\end{OneLiners}
\end{algorithm}

We formulate the main result of this section as follows:

\begin{theorem}\label{thm:zooming-dim}
Consider the \standardMAB. Let \A\ be Algorithm~\ref{alg:nice-1d}. Then
	$\forall\,C>0$
\begin{equation}\label{eq:thm-zooming-dim}
	R_\A(t) \leq O(C \log t)^{1/(2+d)}\,\times t^{1-1/(2+d)}\;\;
	\text{for all $t$},
\end{equation}
where $d$ is the $C$-zooming dimension of the problem instance.
\end{theorem}

\begin{note}{Remark}
The zooming algorithm is \emph{not} parameterized by the $C$ in~\refeq{eq:thm-zooming-dim}, yet satisfies~\refeq{eq:thm-zooming-dim} for all $C>0$. For  sharper guarantees, $C$ can be tuned to the specific problem instance and specific time $t$.
\end{note}

Let us prove Theorem~\ref{thm:zooming-dim}.
Note that after step 1 in Algorithm~\ref{alg:nice-1d} all strategies are covered. (Indeed, if some strategy is activated in step 1 then it covers the entire metric.)
Let $\mu^* = \sup_{u\in X} \mu(u)$ be the maximal expected reward; note that we do not assume that the supremum is achieved by some strategy. Let
	$\Delta(v) = \mu^* - \mu(v)$.
Let us focus on a given phase $i_\Phase$ of the algorithm.

\begin{lemma} \label{lm:bound-active}
If phase $i_\Phase$ is clean then we have
	$\Delta(v) \leq 4\, r_t(v)$
for any time $t$ and any strategy $v$. It follows that
	$n_t(v) \leq O(i_\Phase)\, \Delta^{-2}(v)$.
\end{lemma}

\begin{proof}


Suppose strategy $v$ is played at time $t$. First we claim that
	$I_t(v)\geq \mu^*$.
Indeed, fix $\eps>0$. By definition of $\mu^*$ there exists a strategy $v^*$ such that $\Delta(v^*) < \eps$. Let $v_t$ be an active strategy that covers $v^*$. By the algorithm specification
	$I_t(v) \geq I_t(v_t)$.
Since $v$ is clean at time $t$, by definition of index we have
	$I_t(v_t) \geq \mu(v_t) + r_t(v_t)$.
By the Lipschitz property we have
	$\mu(v_t) \geq \mu(v^*) - L(v_t, v^*)$.
Since $v_t$ covers $v^*$, we have
	$L(v_t, v^*) \leq r_t(v_t)$
Putting all these inequalities together, we have
	$I_t(v) \geq \mu(v^*) \geq \mu^* - \eps$.
Since this inequality holds for an arbitrary $\eps>0$, we in fact have $I_t(v)\geq \mu^*$. Claim proved.

Furthermore, note that by the definitions of ``clean phase'' and ``index''
we have
	$\mu^* \leq I_t(v) \leq \mu(v) + 3\, r_t(v)$
and therefore
	$\Delta(v) \leq 3\, r_t(v)$.

Now suppose strategy $v$ is not played at time $t$. If it has never been played before time $t$ in this phase, then $r_t(v) > 1$ and thus  the lemma is trivial. Else, let $s$ be the last time strategy $v$ has been played before time $t$. Then by definition of the confidence radius
	$r_t(v) = r_{s+1}(v) \geq \sqrt{2/3}\, r_s(v) \geq \tfrac14\, \Delta(v) $.
\end{proof}


\begin{corollary} \label{cor:sparsity}
In a clean phase, for any active strategies $u,v$ we have
	$L(u,v) > \tfrac14 \min(\Delta(u), \Delta(v))$.
\end{corollary}

\begin{proof}
Assume $u$ has been activated before $v$. Let $s$ be the time when $v$ has been activated. Then by the algorithm specification we have
	$L(u,v)> r_s(u)$.
By Lemma~\ref{lm:bound-active}
	$r_s(u)\geq \tfrac14 \Delta(u)$.
\end{proof}

Let $d$ be the the $C$-zooming dimension. For a given time $t$ in the current phase, let $S(t)$ be the set of all strategies that are active at time $t$, and let
	$$A(i,t) = \{ v\in S(t):\; 2^i \leq \Delta^{-1}(v) < 2^{i+1} \}.$$
We claim that	
	$|A(i,t)| \leq \Czoom\, 2^{id}$.
Indeed,  set $A(i,t)$ can be covered by
	$\Czoom\, 2^{id}$
sets of diameter at most $2^{-i}/8$; by Corollary~\ref{cor:sparsity} each of these sets contains at most one strategy from $A(i,t)$.

\begin{claim}\label{cl:clean-phase}
In a clean phase $i_\Phase$, for each time $t$ we have
\begin{align}\label{eq:clean-phase}
 \textstyle{\sum_{v\in S(t)}} \Delta(v)\, n_t(v)
    \leq O(\Czoom\,i_\Phase)^{1-\gamma}\, t^\gamma,
\end{align}
where $\gamma = \tfrac{d+1}{d+2}$ and $d$ is the $C$-zooming dimension.
\end{claim}

\begin{proof}
Fix the time horizon $t$. For a subset $S \subset X$ of strategies, let
	$R_{S} = \sum_{v\in S} \Delta(v)\, n_t(v)$.
Let us choose $\rho\in (0,1)$ such that
$$ \rho t = (\tfrac{1}{\rho})^{d+1} (\Czoom\, i_\Phase)
	=t^{\gamma}\, (\Czoom\,i_\Phase)^{1-\gamma}.
$$

Define $B$ as the set of all strategies $v\in S(t)$ such that $\Delta(v)\leq \rho$.
Recall that by Lemma~\ref{lm:bound-active} for each $v\in A(i,t)$ we have
	$ n_t(v)  \leq O(i_\Phase)\, \Delta^{-2}(v)$.
Then
\begin{align*}
 R_{A(i,t)}
	&\leq O(i_\Phase) \, \textstyle{\sum_{v\in A(i,t)}}\, \Delta^{-1}(v) \\
        & \leq O(2^i\,i_\Phase) \, |A(i,t)|  \\
	& \leq O(\Czoom \,i_\Phase)\, 2^{i(d+1)}
	\\
\sum_{v\in S(t)} \Delta(v)\, n_t(v)
     &\leq
  R_B + \sum_{i< \log(1/\rho)} R_{A(i,t)} \\
	& \leq \rho t + O(\Czoom \, i_\Phase)\, (\tfrac{1}{\rho})^{d+1} \\
	& \leq O \left( \, t^\gamma\, (\Czoom\, i_\Phase)^{1-\gamma} \right).
\qedhere
\end{align*}
\end{proof}

The left-hand side of~\refeq{eq:clean-phase} is essentially the contribution of the current phase to the overall regret. It remains to sum these contributions over all past phases.

\begin{proofof}{Theorem~\ref{thm:zooming-dim}}
Let $i_\Phase$ be the current phase, let $t$ be the time spend in this phase, and let $T$ be the total time since the beginning of phase $1$. Let $R_\Phase(i_\Phase, t)$ be the left-hand side of~\refeq{eq:clean-phase}. Combining Claim~\ref{cl:conf-rad} and Claim~\ref{cl:clean-phase}, we have
\begin{align*}
E[R_\Phase(i_\Phase, t)]
    & < O(\Czoom\,i_\Phase)^{1-\gamma}\, t^\gamma,\\
R_\A(T)
    &= E\left[
    R_\Phase(i_\Phase, t) +
    \sum_{i=1}^{i_\Phase-1} R_\Phase(i, 2^i)
    \right] \\
    & < O(\Czoom\,\log T)^{1-\gamma}\, T^\gamma. \qedhere
\end{align*}
\end{proofof}

\section{Attaining the optimal  per-metric performance}
\label{sec:pmo}

In this section we ask, ``What is the best possible algorithm for the
Lipschitz MAB problem on a given metric space?''  We consider the
\emph{per-metric performance}, which we define as the
worst-case performance of a given algorithm over all possible problem
instances on a given metric. As everywhere else in this
paper, we focus on minimizing the exponent $\gamma$ such that $R_\A(t)
\leq t^\gamma$ for all sufficiently large $t$.  Motivated by the shape of the guarantees in Theorem~\ref{thm:intro.1d}, let us define the \emph{regret dimension} of an algorithm as follows.

\begin{definition}
Consider the Lipschitz MAB problem on a given metric space. For algorithm \A\ and problem instance $\mathcal{I}$ let
$$\RegretDim_{\mathcal{I}}(\A) = \inf_{d\geq 0} \{ \exists t_0\; \forall t\geq t_0\;\;
	R_\A(t) \leq t^{1-1/(d+2)}  \}.$$
The \emph{regret dimension} of \A\ is
	$\RegretDim(\A) = \sup_{\mathcal{I}}\, \RegretDim_{\mathcal{I}}(\A)$,
where the supremum is taken over all problem instances $\mathcal{I}$ on the given metric space.
\end{definition}

Then Theorem~\ref{thm:intro.1d} states that for the Lipschitz MAB problem on $(L_d, [0,1])$, the regret dimension of the ``na\"ive algorithm" is at most $d$. In fact, it is easy to extend the ``na\"ive algorithm" to arbitrary metric spaces. Such algorithm is parameterized by the covering dimension $d$ of the metric space. It divides time into phases of exponentially increasing length, chooses a $\delta$-net during each phase,\footnote{It is easy to see that the cardinality of this $\delta$-net is $K=O(\delta^{-d})$.} and runs a $K$-armed bandit algorithm such as $\textsc{ucb1}$ on the elements of the $\delta$-net. The parameter $\delta$ is tuned optimally given $d$ and the phase length $T$; the optimal value turns out to be $\delta = T^{-1/(d+2)}$. Using the technique from~\cite{Bobby-nips04} it is easy to prove that the regret dimension of this algorithm is at most $d$.

\begin{lemma}\label{lm:naive-alg}
Consider the Lipschitz MAB problem on a metric space $(L,X)$ of covering dimension $d$. Let \A\ be the na\"ive algorithm that uses $\textsc{ucb1}$ in each phase. Then $\RegretDim(\A)\leq d$.
\end{lemma}

\begin{proof}
Let \A\ be the na\"ive algorithm. For concreteness, assume each phase $i$ lasts $2^i$ rounds. By definition of the covering dimension, it suffices to assume that $d$ is a $c$-covering dimension, for some constant $c>0$. By definition of the regret dimension, it suffices to prove that
	$R_\A(t) \leq \tilde{O}(t^\gamma)$
for all $t$, where $\gamma = \tfrac{d+1}{d+2}$. In order to prove \emph{that}, it suffices to show that for each phase $i$ we have
	$R_{(\A,\,i)}(2^i) \leq \tilde{O}(2^{i\gamma}) $,
where $R_{(\A,\,i)}(t)$ is the expected regret accumulated in the first $t$ rounds of phase $i$.

Let us focus on some phase $i$. In this phase the algorithm chooses a $\delta$-net, call it $S$. We claim that $|S|\leq c\,\delta^{-d}$. Indeed, for any $\delta'<\delta$ the metric space can be covered by $c\,\delta^{-d}$ sets of diameter at most $\delta'$, each of which can contain only one point from $S$. Claim proved. The algorithm proceeds to run $\textsc{ucb1}$ on the elements of $S$. By~\cite{bandits-ucb1} the expected regret of $\textsc{ucb1}$ on $K$ arms in $t$ rounds is at most $O(\sqrt{K\, t \log t})$. Since the maximal $\mu$ on $S$ is at most $\delta$ off of the maximal $\mu$ on $X$, we have
$$ R_{(\A,\,i)}(t) \leq O(\sqrt{|S|\, t \log t}) + \delta t
	\leq \tilde{O}(\sqrt{\delta^{-d}\, t} + \delta t).
$$
Plugging in $t=2^i$ and $\delta = t^{-1/(d+2)}$, we obtain
	$R_{(\A,\,i)}(2^i) \leq \tilde{O}(2^{i\gamma}) $
as claimed.
\end{proof}

Thus we ask: {\bf\em is it possible to achieve a better regret dimension}, perhaps using a more sophisticated algorithm? We show that this is indeed the case. Moreover, we provide an algorithm such that for any given metric space its regret dimension is arbitrarily close to optimal.

The rest of this section is organized as follows.  In Section~\ref{sec:PMO-LB} we develop a lower bound on regret dimension. In Section~\ref{sec:fatness-beyond} we will show that for some metric spaces, there exist algorithms whose regret dimension is smaller than the covering dimension.  We develop these ideas further in
Section~\ref{sec:fatness-PMO} and provide an algorithm whose regret dimension is arbitrarily close to optimal.

\OMIT{
One can prove a matching lower bound: $\RegretDim(\A)\geq d$. In fact, this bound applies to any algorithm which uses a generic $k$-armed bandit algorithm on an evenly spaced sample of the strategy space, in the sense made precise in Section~\ref{sec:LB-naive}. In Section~\ref{sec:LB-naive} we prove a matching lower bound for the na\"ive algorithm.

\subsection{Lower bound for the na\"ive-type algorithms}
\label{sec:LB-naive}

We consider a more general version of the na\"ive algorithm described above, abstracting away the crucial idea of using a generic $k$-armed bandit algorithm on an evenly spaced sample of the strategy space.

\begin{definition}
A bandit algorithm is called \emph{uniformly exploring} if it proceeds in phases such that in the beginning of each phase $i$ it chooses $\delta_i>0$, a $\delta_i$-net $\mathcal{N}_i$ of the strategy space, and a $k$-armed bandit algorithm \A, $k=|\mathcal{N}_i|$
\end{definition}

We prove that the covering dimension is the best regret dimension achievable by a na\"ive algorithm.

\begin{theorem}\label{thm:naive-alg}
Consider the na\"ive algorithm for the Lipschitz MAB problem on a metric space of covering dimension $d$. The regret dimension of this algorithm is exactly $d$.
\end{theorem}

\begin{proofof}{\ref{thm:naive-alg}}
\bf{add proof!}
\end{proofof}
} 

\subsection{Lower bound on regret dimension}
\label{sec:PMO-LB}

Let us develop a lower bound on regret dimension of any algorithm on a given metric space. This bound is equal to the covering dimension for highly homogeneous metric spaces (such as those in which all balls of a given radius are isometric to each other), but in general it can be much smaller.

It is known~\cite{bandits-exp3} that a worst-case instance
of the $K$-armed bandit problem consists of $K-1$ strategies
with identical payoff distributions, and one which is slightly
better.  We refer to this as a ``needle-in-haystack''
instance.  The known constructions of lower bounds for Lipschitz MAB
problems rely on creating a \emph{multi-scale} needle-in-haystack
instance in which there are $K$ disjoint open sets, and $K-1$
of them consist of strategies with identical payoff distributions,
but in the remaining open set there are strategies whose payoff
is slightly better.  Moreover, this special open set contains
$K' \gg K$ disjoint subsets, only one of which contains
strategies superior to the others, and so on down through infinitely
many levels of recursion.  To ensure that this construction
can be continued indefinitely, one needs to assume a covering
property which ensures that \emph{each} of the open sets arising
in the construction has sufficiently many disjoint subsets to
continue to the next level of recursion.
\begin{definition} \label{def:mincov}
For a metric space $(L,X)$, we say that $d$ is the
\emph{min-covering dimension} of $X$,
$d = \MinCOV(X)$, if $d$ is the infimum of $\COV(U)$ over all non-empty
open subsets $U \subseteq X$.
The \emph{max-min-covering dimension} of $X$ is defined by
\[ \MaxMinCOV(X) = \sup \{ \MinCOV(Y) \,:\, Y \subseteq X \}.
\]
\end{definition}
The infimum over open $U \subseteq X$ in the
definition of min-covering dimension ensures that every
open set which may arise in the needle-in-haystack
construction described above will contain $\Omega(\delta^{\varepsilon-d})$
disjoint $\delta$-balls
for some sufficiently small $\delta, \varepsilon$.
Constructing lower bounds for Lipschitz MAB algorithms
in a metric space $X$ only requires  that $X$ should have \emph{subsets}
with large min-covering dimension, which explains the
supremum over subsets in the definition of
max-min-covering dimension.

We will use the following simple packing lemma.\footnote{This is a folklore result; we provide the proof for convenience.}

\begin{lemma} \label{lem:packing}
If $Y$ is a metric space of covering dimension
$d$, then for any $b<d$ and $r_0>0$, there exists
$r \in (0,r_0)$ such that $Y$ contains a collection
of at least $r^{-b}$ disjoint open balls of radius
$r$.
\end{lemma}

\begin{proof}
Let $r < r_0$ be a positive number such that every
covering of $Y$ requires more than $r^{-b}$ balls of
radius $2r$.  Such an $r$ exists, because the
covering dimension of $Y$ is strictly greater
than $b$.  Now let $\mathcal{P} = \{B_1,B_2,\ldots,B_M\}$ be any
maximal collection of disjoint $r$-balls.  For
every $y \in Y$ there must exist some ball $B_i
\; (1 \leq i \leq M)$ whose center is within
distance $2r$ of $y$, as otherwise $B(y,r)$
would be disjoint from every element of $\mathcal{P}$
contradicting the maximality of that collection.
If we enlarge each ball $B_i$ to a ball $B_i^+$
of radius $2r$, then every $y \in Y$ is contained
in one of the balls $\{B_i^+ \,|\, 1 \leq i \leq M\}$,
i.e. they form a covering of $Y$.  Hence
$M \geq r^{-b}$ as desired.
\end{proof}

\begin{theorem} \label{thm:mincov-lb}
If $X$ is a metric space and $d$ is the max-min-covering
dimension of $X$ then
$\RegretDim(\A) \geq d$ for every bandit algorithm $\A$.
\end{theorem}

\OMIT{
In fact for every $\gamma < \frac{d+1}{d+2}$, there is
a distribution on instances $\mu$ such that for every
bandit algorithm $\A$ and constant $C < \infty$, it
holds that $\Pr(\forall t \; R_{\A}(t) < C t^{\gamma}) = 0,$
where the probability is over the random choice of $\mu$.
} 

\begin{proof}
Without loss of generality let us assume that $d>0$. Given
    $\gamma < \tfrac{d+1}{d+2},$ let $a < b < c < d$
be such that $\gamma < \tfrac{a+1}{a+2}$. Let $Y$ be a subset of $X$ such that $\MinCOV(Y) \geq c$. Using Lemma~\ref{lem:packing} we recursively construct an infinite sequence of sets $\mathcal{P}_0, \mathcal{P}_1, \ldots$ each consisting of finitely many disjoint open balls in $X$, centered at points of $Y$.  Let $\mathcal{P}_0 = \{X\}$ consist of a single ball that contains all of $X$. If $i>0$, for every  ball $B \in \mathcal{P}_{i-1}$, let $r$ denote the radius of $B$ and choose a number $r_i(B)\in (0, r/4)$ such that $B$ contains $n_i(B) = \lceil r_i(B)^{-b} \rceil$ disjoint balls of radius $r_i(B)$ centered at points of $Y$.  Such a collection of disjoint balls exists, by Lemma~\ref{lem:packing}. Let $\mathcal{P}_i(B)$ denote this collection of disjoint balls and let
    $\mathcal{P}_i = \bigcup_{B \in \mathcal{P}_{i-1}} \mathcal{P}_i(B).$
Now sample a random sequence of balls $B_1, B_2, \ldots$
by picking $B_1 \in \mathcal{P}_1$ uniformly at random,
and for $i>1$ picking $B_i \in \mathcal{P}_i(B_{i-1})$
uniformly at random.

Given a ball $B = B(x_*, r_*)$, let $f_B(x)$ be a Lipschitz
function on $X$ defined by
\begin{align}\label{eq:fB}
f_B(x) =
\begin{cases}
\min \{r_*-L(x,x_*),r_*/2\} & \mbox{if $x \in B$} \\
0 & \mbox{otherwise}
\end{cases}.
\end{align}
Let $f_i = f_{B_i}$ for $i\geq 1$. Define $f_0$ by setting $f_0(x)=1/3$ for all $x \in X$.
The reader may verify that the sum $\mu = \sum_{i=0}^{\infty} f_i$
is a Lipschitz function.  Define the payoff distribution
for $x \in X$ to be a Bernoulli random variable with
expectation $\mu(x)$.  We have thus specified a
randomized construction of an instance $(L,X,\mu)$.

We claim that for any algorithm $\A$ and any constant
$C$,
\begin{align}\label{eq:stronger-result}
    \Pr_{\mu,\,\A}(\forall\, t \; R_{\A}(t) < C t^{\gamma}) = 0.
\end{align}
The proof of this claim is based on a ``needle
in haystack'' lemma (Lemma~\ref{lem:haystack} below)
which states that for all $i$,
conditional on the sequence $B_1,\ldots,B_{i-1}$, with
probability at least $1-O((r_i(B_i))^{(b-a)/2})$, no more than half of the
first $t_i(B_i) = r_i(B_i)^{-a-2}$ strategies picked by $\A$
lie inside $B_i$.  The proof of the lemma is deferred
to the end of this section.

Any strategy $x \not\in B_i$
satisfies $\mu(x) < \mu(x^*) - r_i/2$, so we may conclude
that
\begin{align}
\Pr_\mu\left( R_{\A}(t_i(B_i))
        <\tfrac{1}{4}\, r_i(B_i)^{-a-1} \,|\, B_1,\ldots,B_{i-1} \right)
\leq
    O \left( (r_i(B_i))^{(b-a)/2} \right).
 \label{eq:borel-cantelli}
\end{align}
Denoting $r_i(B_i)$ and $t_i(B_i)$ by $r_i$ and $t_i$, respectively,
we have
$\tfrac14\,r_i^{-a-1} = \tfrac14\,t_i^{(a+1)/(a+2)} > C t_i^{\gamma}$
for all sufficiently large $i$.
As $i$ runs through the positive integers, the terms on the
right side of (\ref{eq:borel-cantelli})
are dominated by a geometric progression because
$r_i(B_i) \leq 4^{-i}.$  By the Borel-Cantelli Lemma,
almost surely there are
only finitely many $i$ such that  the events on the
left side of (\ref{eq:borel-cantelli}) occur.
Thus~\refeq{eq:stronger-result} follows.
\end{proof}

\begin{remark*}
To prove Theorem~\ref{thm:mincov-lb} it suffices to show that for every given algorithm there exists a ``hard'' problem instance. In fact we proved a stronger result~\refeq{eq:stronger-result}: essentially, we construct a probability distribution over problem instances which is hard, almost surely, for every given algorithm. This seems to be the best possible bound since, obviously, a single problem instance cannot be hard for every algorithm.
\end{remark*}


In rest of this subsection we prove the ``needle in haystack'' lemma used in the proof of Theorem~\ref{thm:mincov-lb}.

\begin{lemma} \label{lem:haystack}
Consider the randomized construction of an instance $(L,X,\mu)$ in the proof of Theorem~\ref{thm:mincov-lb}. Fix a bandit algorithm \A. Then for all $i$,
conditional on the sequence $B_1,\ldots,B_{i-1}$, with
probability at least $1-O((r_i(B_i))^{(b-a)/2})$, no more than half of the
first $r_i(B_i)^{-a-2}$ strategies picked by $\A$
lie inside $B_i$.
\end{lemma}

Let us introduce some notation needed to prove the lemma.  Let us fix an arbitrary
Lipschitz MAB algorithm $\A$.  We will assume that $\A$ is
deterministic; the corresponding result for randomized
algorithms follows by conditioning on the algorithm's
random bits (so that its behavior, conditional
on these bits, is deterministic), invoking the lemma
for deterministic algorithms, and then removing the
conditioning by averaging over the distribution of
random bits.
Note that since
our construction uses only $\{0,1\}$-valued payoffs, and
the algorithm $\A$ is deterministic, the entire history
of play in the first $t$ rounds can be summarized by a
binary vector $\sigma \in \{0,1\}^t$, consisting of the
payoffs observed by $\A$ in the first $t$ rounds.
Thus a payoff
function $\mu$ determines a probability distribution
$P_{\mu}$ on the set $\{0,1\}^t$, i.e. the
distribution on $t$-step histories realized when using
algorithm $\A$ on instance $\mu$.

Let $B$ be any ball in the set $\mathcal{P}_{i-1}$,
let
    $n = n_i(B)$, $r = r_i(B)$, and $t=t_i(B) = r_i(B)^{-a-2}$.
Let $B^1,B^2,\ldots,B^n$ be
an enumeration of the balls in $\mathcal{P}_i(B)$.
Choose an arbitrary sequence of balls
$B_1 \supseteq B_2 \supseteq \ldots \supseteq B_{i-1} = B$
such that $B_1 \in \mathcal{P}_1$ and for all $j > 0$
$B_j \in \mathcal{P}(B_{j-1}).$  Similarly, for
$k=1,2,\ldots,n$, choose an arbitrary sequence
of balls $B^k = B^k_{i} \supseteq B^k_{i+1} \supseteq \ldots$
such that $B^k_{j} \in \mathcal{P}(B^k_{j-1})$ for all
$j \geq i.$  Define functions $f_j \; (1 \leq j \leq i-1)$
and $f^k_j \; (j \geq i)$ using the balls $B_j, B^k_j$,
as in the proof of Theorem~\ref{thm:mincov-lb}. Specifically, use definition~\refeq{eq:fB} and set
$f_j = f_{B_j}$ and $f^k_j = f_{B^k_j}$.
Let $\mu^0 = \sum_{j=0}^{i-1} f_j$ and
\begin{align*}
\mu^k &= \mu^0 + \sum_{j=i}^{\infty} f^k_j \quad
\mbox{(for $1 \leq k \leq n$)}.
\end{align*}
Note that the instances $\mu^k \; (1 \leq k \leq n)$ are
equiprobable under our distribution on input instances
$\mu$.  The instance $\mu^0$ is not one that could be
randomly sampled by our construction, but it is useful
as a ``reference measure'' in the following proof.
Note that the functions $\mu^k$ have the following
properties, by construction.
\begin{OneLiners}
\item[(a)] \label{fact:balanced}
$1/3 \leq \mu^k(x) \leq 2/3$ for all $x \in X$.
\item[(b)] \label{fact:close}
$0 \leq \mu^k(x) - \mu^0(x) \leq r$
for all $x \in X$.
\item[(c)] \label{fact:match}
If $x \in X \setminus B^k,$ then $\mu^k(x)=\mu^0(x)$.
\item[(d)] \label{fact:regret}
If $x \in X \setminus B^k,$ then
there exists some point $x^k \in B^k$ such that
$\mu^k(x^k) - \mu^k(x) \geq r/2.$
\end{OneLiners}

Each of the payoff functions $\mu^k \; (0 \leq k \leq n)$
gives rise to a probability distribution $P_{\mu^k}$ on
$\{0,1\}^{t}$ as described in the preceding section.
We will use the shorthand notation $P_k$ instead of
$P_{\mu^k}$.  We will also use $\mathbf{E}_k$ to denote
the expectation of a random variable under distribution
$P_k$.  Finally, we let $N_k$ denote the random variable
defined on $\{0,1\}^t$ that counts the number of rounds
$s \; (1 \leq s \leq t)$ in which algorithm $\A$ chooses
a strategy in $B^k$ given the history $\sigma$.

The following lemma is analogous to Lemma A.1 of~\cite{bandits-exp3},
and its proof is identical to the proof of that lemma.
\begin{lemma} \label{lem:technical}
Let $f \,:\, \{0,1\}^t \rightarrow [0,M]$ be any function defined
on reward sequences $\sigma$.  Then for any $k$,
\[
\mathbf{E}_k[f(\sigma)] \leq
\mathbf{E}_0[f(\sigma)] + \tfrac{M}{2}
\sqrt{- \ln(1-4 r^2) \mathbf{E}_0[N_i] }.
\]
\end{lemma}
Applying Lemma~\ref{lem:technical} with $f = N_k$
and $M = t$, and averaging over $k$, we may apply
exactly the same reasoning as in the proof of
Theorem A.2 of~\cite{bandits-exp3} to derive the
bound
\begin{equation} \label{eq:lousy}
\frac1n \sum_{k=1}^n \mathbf{E}_k(N_k)
\leq
\frac{t}{n} + O \left( tr \sqrt{ \frac{t}{n} } \right).
\end{equation}
Recalling that the actual ball $B_k$ sampled when randomly
constructing $\mu$ in the proof of Theorem~\ref{thm:mincov-lb}
is a uniform random sample from $B^1,B^2,\ldots,B^n$,
we may write $N_*$ to denote the random variable which
counts the number of rounds in which the algorithm plays
a strategy in $B_k$ and the bound (\ref{eq:lousy})
implies
\[
\mathbf{E}(N_*) = O \left( \frac{t}{n} + tr \sqrt{ \frac{t}{n} } \right)
\]
Recalling that $t=r^{-a-2}$ and $n=r^{-b}$, we see
that the $O(tr \sqrt{t/n})$ term
is the dominant term on the right side, and that it is
bounded by $O(t r^{(b-a)/2}).$
An application of Markov's inequality now yields:
$$ \Pr(N_* \geq t/2) = O(r^{(b-a)/2}),$$
completing the proof of Lemma~\ref{lem:haystack}.

\subsection{Beyond the covering dimension}
\label{sec:fatness-beyond}
Thus far, we have seen that every metric space $X$ has a bandit
algorithm $\A$ such that $\RegretDim(\A) = \COV(X)$ (the na\"ive
algorithm), and we have seen (via the needle-in-haystack construction,
Theorem~\ref{thm:mincov-lb})
that $X$ can never have a bandit algorithm satisfying
$\RegretDim(\A) < \MaxMinCOV(X)$.
When $\COV(X) \neq \MaxMinCOV(X)$, which of these two bounds
is correct, or can they both be wrong?

To gain intuition,
we will consider two concrete examples. Consider an infinite rooted tree
where for each level $i\in\N$ most nodes have out-degree $2$, whereas
the remaining nodes (called \emph{fat nodes}) have out-degree $x>2$ so
that the total number of nodes is $4^i$. In our first example, there
is exactly one fat node on every level and the fat nodes form a path
(called the \emph{fat leaf}). In our second example, there are exactly
$2^i$ fat nodes on every level $i$ and the fat nodes form a binary
tree (called the \emph{fat subtree}).
In both examples, we assign a \emph{weight} of $2^{-id}$ (for
some constant $d>0$) to each level-$i$ node; this weight encodes
the diameter of the set of points contained in the corresponding
subtree.  An infinite rooted
tree induces a metric space $(L,X)$ where $X$ is the set of all
infinite paths from the root, and for $u,v \in X$ we define
$L(u,v)$ to be the weight of the least common ancestor of paths
$u$ and $v$.
In both examples, the covering dimension is $2d$, whereas the
max-min-covering dimension is only $d$ because the ``fat subset''
(i.e. the fat leaf or fat subtree) has covering dimension at most
$d$, and every point outside the fat subset has an open neighborhood
of covering dimension $d$.

\OMIT{ 
In the next few paragraphs, we sketch
some algorithms for dealing with certain metric spaces that have fat
subsets, as a means of building intuition leading up to the
(rather complicated) optimal algorithm for general metric spaces.
The gory details are omitted until we reach the description of
the algorithm for general metric spaces.
} 

In both of the metrics described above,
the zooming algorithm (Algorithm~\ref{alg:nice-1d})
performs poorly when the optimum $x^*$ is located
inside the fat subset $S$, because it is too burdensome to keep
covering\footnote{Recall that a strategy $u$ is called
  \emph{covered} at time $t$ if for some active strategy $v$ we have
  $L(u,v) \leq r_t(v)$.}
the profusion of strategies located near $x^*$ as the
ball containing $x^*$ shrinks.  An improved algorithm,
achieving regret exponent $d$,  modifies the zooming algorithm
by imposing \emph{quotas} on the number of
active strategies that lie outside $S$.  At any
given time, some strategies outside $S$ may not be
covered; however, it is guaranteed that there exists
an optimal strategy which eventually becomes covered
and remains covered forever afterward.
Intuitively, if some optimal strategy lies in $S$ then
imposing a quota on active strategies outside $S$ does not hurt.
If no optimal strategy lies in $S$ then all of $S$
gets covered eventually and stays covered thereafter, in which case
the uncovered part of the strategy set has low covering dimension and
(starting after the time when $S$ becomes permanently covered)
no quota is ever exceeded.

This use of quotas extends to the following general setting which
abstracts the idea of ``fat subsets'':
\begin{definition}\label{def:fatness-dim}
Fix a metric space $(L,X)$. A closed subset $S\subset X$ is \emph{$d$-fat} if
	$\COV(S)\leq d$
and for any open superset $U$ of $S$ we have
	$\COV(X\setminus U) \leq d$.
More generally, a \emph{$d$-fat decomposition} of depth $k$ is a decreasing sequence
	$X = S_0 \supset \ldots \supset S_k \supset S_{k+1} = \emptyset$
of closed subsets such that
	$\COV(S_k)\leq d$ and $\COV(S_i\setminus U) \leq d$
whenever $i \in [k]$ and $U$ is an open superset of $S_{i+1}$.
\end{definition}

\begin{example}
Let $(L,X)$ be the metric space in either of the two
``tree with a fat subset'' examples. Then the corresponding
``fat subset'' $S$ is $d$-fat. For an example of a fat
decomposition of depth $k=2$, consider the product
metric $(L^*, X\times X)$ defined by
	$$L^*( (x_1, x_2),(y_1, y_2)) = L(x_1,y_1) + L(x_2, y_2),$$
with a fat decomposition given by
	$S_1 = (S\times X) \cup (X\times S)$
and $S_2 = S\times S$.
\end{example}

When $X$ is a metric space with a $d^*$-fat decomposition
$\mathcal{D}$,
the algorithm described earlier can be modified to achieve
regret $O \left(t^{\gamma} \right)$
for any $\gamma > 1-1/(d^*+2)$, by instituting a separate
quota for each subset $S_i.$  The algorithm requires access
to a \emph{$\mathcal{D}$-covering oracle} which for a given
$i$ and a given finite set of open balls (given by the
centers and the radii) either reports that the balls
cover $S_i$, or returns some strategy in $S_i$ which is
not covered by the balls.  No further knowledge of
$\mathcal{D}$ or the metric space is required.

\begin{theorem}\label{thm:fatness}
Consider the Lipschitz MAB problem on a fixed compact metric space with a $d^*$-fat decomposition $\mathcal{D}$. Then for any $d>d^*$ there is an algorithm $\A_\mathcal{D}$ such that
	$\RegretDim(\A_\mathcal{D}) \leq d$.
\end{theorem}

\begin{note}{Remarks.}

{\bf (1)} We can relax the compactness assumption in Theorem~\ref{thm:fatness}: instead, we can assume that the \emph{completion} of the metric space is compact and re-define the sets in the $d$-fat decomposition as subsets of the completion (possibly disjoint with the strategy set). This corresponds to the ``fat leaf'' which lies outside the strategy set. Such extension requires some minor modifications.

{\bf (2)} The per-metric guarantee expressed by Theorem~\ref{thm:fatness} can be complemented with sharper \emph{per-instance} guarantees. First, for every problem instance $\mathcal{I}$ the per-instance regret dimension
	$\RegretDim_{\mathcal{I}}(\A)$
is upper-bounded by the zooming dimension of $\mathcal{I}$. Second, if for some $c>0$ the $c$-covering dimension of $X$ is finite then for some $\gamma<1$ and all $t$ we have
	$R_\A(t) \leq O(c\, t^{\gamma})$.
However, as this extension is tangential to our main storyline, we focus on analyzing the regret dimension.
\end{note}

\newcommand{\PhaseAlg}{\ensuremath{\A_{\mathtt{ph}}}}

\xhdr{The algorithm.}
Our algorithm proceeds in phases $i=1,2,3,\,\ldots\;$ of $2^i$ rounds each. In a given phase, we run a fresh instance of the following \emph{phase algorithm}
	$\PhaseAlg(T,d,\mathcal{D})$
parameterized by the phase length $T = 2^i$, \emph{target dimension} $d>d^*$ and the $\mathcal{D}$-covering oracle. The phase algorithm is a version of a single phase of the zooming algorithm (Algorithm~\ref{alg:nice-1d}) with very different rules for activating strategies. As in Algorithm~\ref{alg:nice-1d}, the confidence radius and the index are defined by~\refeq{eq:confidence-radius} and~\refeq{eq:index}, respectively.  At the start 
of each round some strategies are activated, and then 
an active strategy with the maximal index is played.

Let us specify the activation rules. Let $k$ be the depth of the decomposition $\mathcal{D}$, and denote
	$\mathcal{D} = \{S_i\}_{i=0}^{k+1}$.
Initially the algorithm  constructs $2^{-j}$-nets $\mathcal{N}_j$, $j\in\N$,
using the covering oracle.  It
finds the largest $j$ such that
	$\mathcal{N} = \mathcal{N}_j$
contains at most $\tfrac12\; T^{d/(d+2)}$ points, and activates all strategies in $\mathcal{N}$. The rest of the active strategies are partitioned into $k+1$ pools $P_i \subset S_i$ such that at each time $t$ each pool $P_i$ satisfies the following \emph{quota} (that we denote $Q_i$):
\begin{equation}\label{eq:fatness-quotas}
|\{ u\in P_i:\, r_t(u) \geq \rho \}| \leq C_\rho\; \rho^{-d}\quad
\end{equation}
where $\rho = T^{- 1/(d+2)}$ and $C_\rho = (64k\,\log \tfrac{1}{\rho})^{-1}$.
In the beginning of each round the following activation routine is performed. If there exists a set $S_i$ such that some strategy in $S_i$ is not covered and \emph{there is room under the corresponding quota $Q_i$}, pick one such strategy, activate it, and add it to the corresponding pool $P_i$. Since for a given strategy $u$ the confidence radius $r_t(u)$ is non-increasing in $t$, the constraint~\refeq{eq:fatness-quotas} is never violated. Repeat until there are no such sets $S_i$ left.  This completes the description of the algorithm.

\xhdr{Analysis.}
As 
was the case in Section~\ref{sec:adaptive-exploration},
the analysis of the unbounded-time-horizon algorithm reduces to
proving a lemma about the regret of each phase algorithm.

\newcommand{\TMin}{\ensuremath{t_{\text{min}}}}

\begin{lemma}\label{lm:fatness-phase}
Fix a problem instance in the setting of Theorem~\ref{thm:fatness}. Let
	$\PhaseAlg(T) = \PhaseAlg(T, d,\mathcal{D})$.
Then
\begin{equation}\label{eq:fatness-phase}
(\exists\,  \TMin<\infty)\; (\forall T\geq \TMin)\quad
	R_{\PhaseAlg(T)} (T) \leq T^{1-1/(d+2)}.
\end{equation}
\end{lemma}
Note that the lemma bounds the regret of $\PhaseAlg(T)$ for time $T\geq \TMin$ only.
Proving Theorem~\ref{thm:fatness} is now straightforward:

\begin{proofof}{Theorem~\ref{thm:fatness}}
Let $\PhaseAlg(T)$ be the phase algorithm from Lemma~\ref{lm:fatness-phase}.
Recall that in each phase $i$ in the overall algorithm \A\ we simply run a fresh instance of algorithm $\PhaseAlg(2^i)$ for $2^i$ steps.

Let $t_0$ be the $\TMin$ from~\refeq{eq:fatness-phase} rounded up to the nearest end-of-phase time. Let $i_0$ be the phase starting at time $t_0+1$. Note that
	$R_\A(t_0) \leq t_0$.
Let $R_i$ be the regret accumulated by \A\ during phase $i$. Let $\gamma = \tfrac{d+1}{d+2}$.
Then for any time
	$t\geq t_0^{1/\gamma}$
in phase $i$ we have
$ R_\A(t) \leq t_0 + \sum_{j=i_0}^{i} R_j
	\leq t_0 + \sum_{j=i_0}^{i}  (2^j)^\gamma
	\leq O(t^\gamma).
$
\end{proofof}

In the remainder of this section we prove Lemma~\ref{lm:fatness-phase}. Let us fix a problem instance of the  Lipschitz MAB problem on a compact metric space $(L,X)$ with a depth-$k$ $d^*$-fat decomposition
	$\mathcal{D} = \{S_i\}_{i=0}^{k+1}$.
Fix $d>d^*$ and let
	$\PhaseAlg(T) = \PhaseAlg(T, d,\mathcal{D})$
be the phase algorithm. Let $\mu$ be the expected reward function and let
	$\mu^* = \sup_{u\in X} \mu(u)$
be the optimal reward. Let $\Delta(u) = \mu^* - \mu(u)$.

By definition of the Lipschitz MAB problem, $\mu$ is a continuous function on the metric space $(L,X)$. Therefore the supremum $\mu^*$ is achieved by some strategy (call such strategies \emph{optimal}). Say that a run of algorithm $\PhaseAlg(T)$ is \emph{well-covered} if at every time $t\leq T$ some optimal strategy is covered.

Say that a run of algorithm $\PhaseAlg(T)$ is \emph{clean} if the property in Claim~\ref{cl:conf-rad} holds for all times $t\leq T$. Note that a given run is clean with probability at least $1-T^{-2}$. The following lemma adapts the technique from Lemma~\ref{lm:bound-active} to the present setting:

\begin{claim}\label{lm:fatness-confRad}
Consider a clean run of algorithm $\PhaseAlg(T)$.
\begin{OneLiners}
\item[(a)] If strategies $u,v$ are active at time $t\leq T$ then
	$\Delta(v) - \Delta(u) \leq 4 r_t(v)$.
\item[(b)] if the run is well-covered and strategy $v$ is active at time $t\leq T$ then
	$\Delta(v) \leq 4 r_t(v)$.
\end{OneLiners}
\end{claim}

The quotas~\refeq{eq:fatness-quotas} are chosen so that the regret computation in Claim~\ref{cl:clean-phase} works out for a clean and well-covered run of algorithm $\PhaseAlg(T)$.

\begin{claim}\label{lm:fatness-regret}
	$R_\A (T) \leq T^{1-1/(d+2)}$
for any clean well-covered run of algorithm $\A = \PhaseAlg(T)$.
\end{claim}

\begin{proof}[ Sketch]
Let  $A_t(\delta)$ be the set of all strategies $u\in X$ such that $u$ is active at time $t\leq T$ and
	$\delta \leq r_t(u) <2\delta$.
Note that for any such strategy we have
	$n_t(u) \leq O(\log T)\, \delta^{-2}$
and
	$\Delta(u) \leq 4 r_t(u) < 8 \delta$.
Write
\begin{align*}
R^*(T) & := \textstyle{\sum_{u\in X}} \Delta(u)\, n_T(u)
	 \leq \rho T + \textstyle{
		\sum_{i=0}^{\cel{\log 1/\rho}}
		\sum_{u\in A_T(2^{-i})} \Delta(u)\, n_T(u)},
\end{align*}
where $\rho = T^{- 1/(d+2)}$ and apply the quotas~\refeq{eq:fatness-quotas}.
\end{proof}

Let $S_{\ell}$ be the smallest set in $\mathcal{D}$ which contains some optimal strategy. Then there is an optimal
strategy contained in $S_{\ell} \setminus S_{\ell+1}$; let $u^*$ be one such strategy. The following claim essentially shows that the irrelevant high-dimensional subset $S_{\ell+1}$ is eventually pruned away.

\begin{claim}\label{lm:fatness-entombment}
There exists an open set $U$ containing $S_{\ell+1}$ such that $u^*\not\in U$ and $U$ is always covered throughout the first $T$ steps of any clean run of algorithm $\PhaseAlg(T)$, provided that $T$ is sufficiently large.
\end{claim}

\begin{proof}
$S_{\ell+1}$ is a compact set since it is a closed subset of a compact metric space. Since function $\mu$ is continuous, it assumes a maximum value on $S_{\ell+1}$. By construction, this maximum value is strictly less than $\mu^*$. So there exists $\eps>0$ such that $\Delta(w)> 8 \eps$ for any $w\in S_{\ell+1}$. Define
	$U = B(S_{\ell+1}, \eps/2)$.
Note that $u^* \not\in U$ since
	$ 8\eps < \Delta(w) \leq L(u^*,w)$
for any $w\in S_{\ell+1}$.

Recall that in the beginning of algorithm $\A(T)$ all strategies in some $2^{-j}$-net $\mathcal{N}$ are activated. Suppose $T$ is large enough so that $2^{-j} \leq \eps$.

Consider a clean run of algorithm $\PhaseAlg(T)$.  We claim that $U$ is covered at any given time $t\leq T$. Indeed, fix $u\in U$. By definition of $U$ there exists a strategy
	$w\in S_{\ell+1}$ such that $L(u,w) < \eps/2$.
By definition of $\mathcal{N}$ there exist
	$v,v^*\in \mathcal{N}$
such that $L(v,w) \leq \eps$ and $L(u^*,v^*) \leq \eps$.
Note that:
\begin{OneLiners}
\item[(a)] $\Delta(v^*) = \mu(u^*) - \mu(v^*) \leq  L(u^*, v^*) \leq \eps$.
\item[(b)] Since $L(v,w) \leq \eps$ and $\Delta(w) > 8\eps$, we have $\Delta(v) > 7\eps$.
\item[(c)] By Claim~\ref{lm:fatness-confRad} we have
	$\Delta(v) - \Delta(v^*) \leq 4 r_t(v^*)$.
\end{OneLiners}
Combining (a-c), it follows that $r_t(v) \geq \tfrac32\,\eps \geq L(u,v)$, so $v$ covers $u$. Claim proved.
\end{proof}

\begin{proofof}{Lemma~\ref{lm:fatness-phase}}
By Claim~\ref{lm:fatness-regret} it suffices to show that if $T$ is sufficiently large then any clean run of algorithm $\PhaseAlg(T)$ is well-covered.
(Runs that are not clean contribute only $O(1/T)$ to the expected
regret of $\PhaseAlg(T)$, because the probability that a run is not
clean is at most $T^{-2}$ and the regret of such a run is at most $T$.)
Specifically, we will show that $u^*$ is covered at any time $t\leq T$
during a clean run of $\PhaseAlg(T)$. It suffices to show that at any time $t\leq T$ there is room under the corresponding quota $Q_{\ell}$ in~\refeq{eq:fatness-quotas}.

Let $U$ be the open set from Claim~\ref{lm:fatness-entombment}. Since $U$ is an open neighborhood of $S_{\ell+1}$, by definition of the fat decomposition it follows that
	$\COV(S_{\ell} \setminus U)\leq d^*$.
Define $\rho$ and $C_\rho$ as in~\refeq{eq:fatness-quotas} and fix $d'\in (d^*, d)$. Then for any sufficiently large $T$ it is the case that (i) $S_{\ell} \setminus U$ can be covered with
	$(\tfrac{1}{\rho})^{d'}$
sets of diameter $<\rho$ and moreover (ii) that
	$(\tfrac{1}{\rho})^{d'} \leq \tfrac12\,C_\rho\, \rho^{-d}$.

Fix time $t\leq T$ and let $A_t$ be the set of all strategies $u$ such that $u$ is in the pool $P_{\ell}$ at time $t$ and
	$r_t(u)\geq \rho$.
Note that $A_t \subset S_{\ell}\setminus U$  since $U$ is always covered, and by the specification of $\PhaseAlg$ only active uncovered strategies in $S_{\ell}$ are added to pool $P_{\ell}$. Moreover, $A_t$ is $\rho$-separated. (Indeed, let $u,v\in A_t$ and assume $u$ has been activated before $v$. Then
	$L(u,v) > r_s(u) \geq r_t(u) \geq \rho$,
where $s$ is the time when $v$ was activated.) It follows that
	$|A_t| \leq \tfrac12\,C_\rho\, \rho^{-d}$,
so there is room under the corresponding quota $Q_{\ell}$ in~\refeq{eq:fatness-quotas}.
\end{proofof}

\subsection{The per-metric optimal algorithm}
\label{sec:fatness-PMO}

The algorithm in Theorem~\ref{thm:fatness} requires a fat decomposition of finite depth, which in general might not exist. To extend the ideas of the preceding section to arbitrary metric
spaces, we must generalize Definition~\ref{def:fatness-dim} to \emph{transfinitely infinite} depth.
\begin{definition}\label{def:fatness-transfinite}
Fix a metric space $(L,X)$. Let $\beta$ denote an arbitrary ordinal.
A \emph{transfinite $d$-fat decomposition} of depth $\beta$
is a  transfinite sequence
	$\{S_\lambda\}_{0 \leq \lambda \leq \beta}$
of closed subsets of $X$ such that:
\begin{OneLiners}
\item[(a)] $S_0 = X$, $S_\beta = \emptyset$, and
	$S_\nu \supseteq S_\lambda$ whenever $\nu < \lambda$.
\item[(b)] if $V\subset X$ is closed, then the set
    $\{\text{ordinals } \nu \leq \beta$:\, $V \mbox{ intersects } S_\nu \}$
has a maximum element.
\item[(c)] for any ordinal $\lambda \leq \beta$ and any open set
$U\subset X$ containing $S_{\lambda+1}$ we have
	$\COV(S_\lambda \setminus U) \leq d$.
\end{OneLiners}
\end{definition}
Note that for a finite depth $\beta$ the above definition is
equivalent to Definition~\ref{def:fatness-dim}.  In
Theorem~\ref{thm:PMO} below, we will show how to modify
the ``quota algorithms'' from the previous section to
achieve regret dimension $d$ in any metric with a transfinite
$d^*$-fat decomposition for $d^* < d$.  This gives
an optimal algorithm for every metric space $X$ because
of the following
surprising relation between the max-min-covering dimension
and transfinite fat decompositions.

\begin{proposition} \label{prop:fatness-dim}
For every compact metric space $(L,X)$, the
max-min-covering dimension of $X$ 
is equal to the infimum of all
$d$ such that $X$ has a transfinite
$d$-fat decomposition.
\end{proposition}
\begin{proof}
\newcommand{\ThinPt}{{\operatorname{TP}}}
\newcommand{\FatPt}{{\operatorname{FP}}}
If $\emptyset \neq Y \subseteq X$ and $\MinCOV(Y) > d$ then, by
transfinite induction, $Y \subseteq S_{\lambda}$ for
all $\lambda$ in any transfinite $d$-fat decomposition,
contradicting the fact that $S_{\beta} = \emptyset$.  Thus,
the existence of a transfinite $d$-fat decomposition of $X$
implies $d \geq \MaxMinCOV(X)$.  To complete the proof we
will construct, given any $d > \MaxMinCOV(X)$,  a transfinite
$d$-fat decomposition of depth $\beta$, where $\beta$ is
any ordinal whose  cardinality exceeds that of $X$.
For a metric space $Y$, define the set of \emph{$d$-thin points}
$\ThinPt(Y,d)$
to be the union of all open sets $U \subseteq Y$ satisfying
$\COV(U) < d$. Its complement, the set of \emph{$d$-fat points},
is denoted by $\FatPt(Y,d)$. Note that it
is a closed subset of $Y$.

For an ordinal $\lambda \leq \beta$, we define a set $S_{\lambda}$
using transfinite induction as follows:
\begin{OneLiners}
\item[1.] $S_0 = X$ and  $S_{\lambda+1} = \FatPt(S_{\lambda},d)$
for each ordinal $\lambda$.
\item[2.]  If $\lambda$ is a limit ordinal then
	$S_{\lambda} = \bigcap_{\nu < \lambda} S_\nu$.	
\end{OneLiners}
Note that each $S_\lambda$ is closed, by transfinite induction.
It remains to show that
	$\mathcal{D} = \{S_\lambda\}_{\lambda \in \mathcal{O}}$
satisfies the properties (a-c) in
Definition~\ref{def:fatness-transfinite}.
It follows immediately from the construction that
$S_0 = X$ and $S_{\nu} \supseteq S_{\lambda}$ when
$\nu < \lambda.$  To prove that $S_{\beta} = \emptyset$,
observe first that the sets $S_{\lambda} \setminus S_{\lambda+1} \;
(\mbox{for } 0 \leq \lambda < \beta)$ are disjoint subsets of $X$,
and the number of such sets is greater than the cardinality of $X$,
so at least one of them is empty.  This means that
$S_{\lambda} = S_{\lambda+1}$ for some $\lambda < \beta.$
If $S_{\lambda} = \emptyset$ then $S_{\beta} = \emptyset$ as
desired.  Otherwise, the relation $\FatPt(S_{\lambda},d) =
S_{\lambda}$ implies that $\MinCOV(S_{\lambda}) \geq d$
contradicting the assumption that $\MaxMinCOV(X) < d.$
This completes the proof of property (a).  To prove
property (b), suppose $\{\nu_i \,|\, i \in \mathcal{I}\}$
is a set of ordinals such that $S_{\nu_i}$ intersects $V$
for every $i$.  Let $\nu = \sup \{\nu_i\}.$  Then
$S_{\nu} \cap V = \bigcap_{i \in \mathcal{I}} (S_{\nu_i} \cap V)$,
and the latter set is nonempty because $X$ is compact and the
closed sets $\{S_{\nu_i} \cap V \,|\, i \in \mathcal{I}\}$
have the finite intersection property.  Finally, to prove
property (c), note that if $U$ is an open neighborhood of
$S_{\lambda+1}$ then the set $T = S_{\lambda} \setminus U$ is
closed (hence compact) and is contained in $\ThinPt(S_{\lambda},d)$.
Consequently $T$ can be covered by open sets $V$ satisfying
$\COV(V) < d$.  By compactness of $T$, this covering has a
finite subcover $V_1,\ldots,V_m$, and consequently
$\COV(T) = \max_{1 \leq i \leq m} \COV(V_i) < d.$
\end{proof}

\begin{theorem}\label{thm:PMO}
Consider the Lipschitz MAB problem on a compact metric space $(L,X)$.
For any $d>\MaxMinCOV(X)$ there exists an algorithm $\A_d$ such that
	$\RegretDim(\A_d) \leq d$.
\end{theorem}
Note that Theorem~\ref{thm:intro.pmo} follows immediately by
combining Theorem~\ref{thm:PMO} with Theorem~\ref{thm:mincov-lb}.

\newcommand{\DpthOracle}{{\mathtt{Depth}}}
\newcommand{\DCovOracle}{{\mbox{$\mathcal{D}$-$\mathtt{Cov}$}}}
\newcommand{\ve}{\varepsilon}

We next describe an algorithm $\A_d$ satisfying Theorem~\ref{thm:PMO}.
The algorithm requires two oracles: a depth oracle $\DpthOracle(\cdot)$
and a $\mathcal{D}$-covering oracle $\DCovOracle(\cdot)$.
For any finite set of open balls $B_0, B_1, \ldots, B_n$
(given via the centers and the radii) whose union is denoted
by $B$, $\DpthOracle(B_0,B_1,\ldots,B_n)$ returns the maximum ordinal
$\lambda$ such that $S_{\lambda}$ intersects the closure
$\overline{B}$; such an ordinal exists
by Definition~\ref{def:fatness-transfinite}(b).\footnote{To
avoid the question of how arbitrary ordinals are represented
on the oracle's output tape, we can instead say that the
oracle outputs a point $u \in S_{\lambda}$ instead of outputting
$\lambda.$  In this case, the definition of $\DCovOracle$
should be modified so that its first argument is a point
of $S_{\lambda}$ rather than $\lambda$ itself.}
Given a finite set of open balls $B_0, B_1, \ldots, B_n$
with union $B$ as above, and an ordinal $\lambda$,
$\DCovOracle(\lambda,B_0,B_1,\ldots,B_n)$ either reports that
$B$ covers $S_{\lambda}$, or it returns a strategy
	$x \in S_{\lambda} \setminus B.$

\xhdr{The algorithm.}
Our algorithm proceeds in phases $i=1,2,3,\ldots$
of $2^i$ rounds each.  In any given phase $i$, there
is a ``target ordinal'' $\lambda(i)$ (defined at the end of the
preceding phase), and we run an algorithm during the
phase which: (i) activates some nodes initially;
(ii) plays a version of the zooming algorithm which
only activates strategies in $S_{\lambda(i)}$;
(iii) concludes the phase by computing $\lambda(i+1)$.
The details are as follows.  In a given phase we run a fresh instance
of a phase algorithm $\PhaseAlg(T,d,\lambda)$ where $T=2^i$
and $\lambda = \lambda(i)$
is a \emph{target ordinal} for phase $i$, defined below when
we give the full description of $\PhaseAlg(T,d,\lambda).$
The goal of $\PhaseAlg(T,d,\lambda)$ is to satisfy the
per-phase bound
\begin{equation} \label{eq:phasealg}
R_{\PhaseAlg(T,d,\lambda)}(T) = \widetilde{O}(T^\gamma)
\end{equation}
for all $T > T_0$, where $\gamma = 1-1/(d+2)$ and $T_0$ is a
number which may depend on the instance $\mu$.
Then, to derive the bound $R_{\A_d}(t) = \widetilde{O}(t^\gamma)$
for all $t$ we simply sum per-phase bounds over all phases
ending before time $2t$.

Initially $\PhaseAlg(T,d,\lambda)$ uses the covering oracle to
construct $2^{-j}$-nets $\mathcal{N}_j$, $j = 0,1,2,\ldots$, until
it finds the largest $j$ such that
	$\mathcal{N} = \mathcal{N}_j$
contains at most $\tfrac12\; T^{d/(d+2)} \log(T)$ points.
It  activates
all strategies in $\mathcal{N}$ and
sets $$\ve(i) = \max\{2^{-j}, 32 \, T^{-1/(d+2)} \log(T)\}.$$
After this initialization
step, for every active strategy $v$ we define the confidence
radius
$$r_t(v) := \max \left\{ T^{-1/(d+2)},
\sqrt{ \frac{8 \log T}{2 + n_t(v)} } \right\},$$ where
$n_t(v)$ is the number of times $v$ has been played by
the phase algorithm $\PhaseAlg(T,d,\lambda)$ before time $t$.
Let $B_0, B_1, \ldots, B_n$ be an enumeration of
the open balls belonging to the collection
$$\{B(v,r_t(v)) \;|\; v \mbox{ active at time $t$}\}.$$
If $n < \tfrac12\; T^{d/(d+2)} \log(T)$
then we perform the oracle call
$\DCovOracle(\lambda,B_0,\ldots,B_n),$
and if it reports that a point $x \in S_{\lambda}$
is uncovered, we activate $x$ and set $n_t(x)=0.$
The index of an active strategy $v$ is
defined as $\mu_t(v) + 4 r_t(v)$ --- note the
slight difference from the index defined
in Algorithm~\ref{alg:nice-1d} --- and we always play the
active strategy with maximum index.
To complete the description of the algorithm, it remains
to explain how the ordinals $\lambda(i)$ are defined.  The
definition is recursive, beginning with $\lambda(1) = 0$.
At the end of phase $i \; (i \geq 1)$, we let
$B_0, B_1, \ldots, B_m$ be an enumeration of the
open balls in the set $\{B(v,\ve(i)) \,|\,
v \mbox{ active}, \, r_T(v) < \ve(i)/2 \}.$
Finally, we set $\lambda(i+1) = \DpthOracle(B_0,B_1,\ldots,B_m).$

\begin{proofof}{Theorem~\ref{thm:PMO}}
Since we have modified the definition of index, we must
prove a variant of Claim~\ref{lm:fatness-confRad} which
asserts the following:
\begin{align}\label{eq:fatness-confRad}
\text{In a clean run of $\PhaseAlg$, if
$u,v$ are active at time $t$ then $\Delta(v) - \Delta(u) \leq
5 r_t(v).$ }
\end{align}
To prove it, let $s$ be the latest round in
$\{1,2,\ldots,t\}$ when $v$ was played.  We have
$r_t(v) = r_s(v)$, and $\Delta(v) - \Delta(u) =
\mu(u) - \mu(v)$, so it remains to prove that
\begin{equation} \label{eqn:idx5.1}
\mu(u) - \mu(v) \leq 5 r_s(v).
\end{equation}
From the fact that $v$ was played instead of $u$
at time $s$, together with the fact that both
strategies are clean,
\begin{align}
\label{eqn:idx5.2}
\mu_s(u) + 4 r_s(u) & \leq \mu_s(v) + 4 r_s(v) \\
\label{eqn:idx5.3}
\mu(u) - \mu_s(u) & \leq  r_s(u) \\
\label{eqn:idx5.4}
\mu_s(v) - \mu(v) & \leq  r_s(v).
\end{align}
We obtain (\ref{eqn:idx5.1}) by adding
(\ref{eqn:idx5.2})-(\ref{eqn:idx5.4}),
noting that $r_s(u) > 0$. This completes the proof of~\refeq{eq:fatness-confRad}.

Let $\lambda$ be the maximum ordinal such that
$S_{\lambda}$ contains an optimal strategy $u^*$; such an
ordinal exists by Definition~\ref{def:fatness-transfinite}(b).
We will prove that for sufficiently
large $i$, if the $i$-th phase is clean, then
$\lambda(i) = \lambda$.
The set $S_{\lambda+1}$ is compact, and the function $\mu$ is
continuous, so it assumes a maximum value on $S_{\lambda+1}$
which is, by construction, strictly less than $\mu^*$.
Choose $\ve>0$ such that $\Delta(w) > 5 \ve$ for all
$w \in S_{\lambda+1}$, and choose $T_0 = 2^{i_0}$ such
that $\ve(i_0) \leq \ve.$  We shall prove that for all
$T = 2^i \geq T_0$ and all ordinals $\nu$, a clean run
of $\PhaseAlg(T,d,\nu)$ results in setting $\lambda(i+1)=\lambda$.
First, let $v^* \in \mathcal{N}$ be such that
$L(u^*,v^*) \leq \ve(i).$  If $v$ is active and
$r_T(v) < \ve(i)/2$ then~\refeq{eq:fatness-confRad}
implies that $\Delta(v) - \Delta(v^*) \leq \tfrac52\, \ve(i)$
hence $\Delta(v) \leq \tfrac72\, \ve(i)$.  As $\Delta(w) > 5 \ve
\geq 5 \ve(i)$ for all $w \in S_{\lambda+1}$, it follows
that the closure of $B(v,\ve(i))$ does not intersect
$S_{\lambda+1}.$  This guarantees that
$\DpthOracle(B_0,B_1,\ldots,B_m)$ returns
an ordinal less than or equal to $\lambda.$
Next we must prove that this ordinal is greater
than or equal to $\lambda.$
Note that the
total number of strategies activated by $\PhaseAlg(T,d,\nu)$
is bounded above by $T^{d/(d+2)} \log(T)$.  Let
$A_T$ denote the set of strategies active at time $T$ and let
\begin{align*}
v^0 &= \arg \max_{v \in A_T} n_T(v).
\end{align*}
By the pigeonhole principle,
$n_T(v^0) \geq T^{2/(d+2)} / \log(T)$
and hence $r_T(v^0) < 3 T^{-1/(d+2)} \log(T).$
If $t$ denotes the last time at which
$v^0$ was played, then we have
\begin{align*}
I_t(v^0) &= \mu_t(v^0) + 4 r_t(v^0)
 \leq \mu^* + 5 r_t(v^0) \\
& \leq \mu^* + 15 T^{-1/(d+2)} \log(T)
 < \mu^* + \ve(i)/2,
\end{align*}
provided that the phase is clean and that
$T \geq T_0.$  Since $v^0$ had maximum index
at time $t$, we deduce that $I_t(v^*) < \mu^* + \ve(i)/2$
as well.  As $L(u^*,v^*) \leq \ve(i)$ we have
$\mu_t(v^*) \geq \mu^* - \ve(i) - r_t(v^*)$ provided
the phase is clean.  To finish the proof we observe
that
\[
\mu^* + \ve(i)/2 > I_t(v^*) \geq \mu^* - \ve(i) + 3 r_t(v^*)
\]
which implies $r_t(v^*) < \ve(i)/2$.  Since the confidence
radius does not increase over time, we have $r_T(v^*) < \ve(i)/2$
so $B(v^*,\ve(i))$ is one of the balls $B_0,B_1,\ldots,B_m.$
Since $u^*$ is contained in the closure of this ball, we
may conclude that $\DpthOracle(B_0,B_1,\ldots,B_m)$ returns
the ordinal $\lambda$ as desired.

Let $U = B(S_{\lambda+1}, \ve(i)/2)$.  As in
Claim~\ref{lm:fatness-entombment} it holds
that in any clean phase, $U$ is covered throughout
the phase by balls centered at points of $\mathcal{N}$.
Hence for any pair of consecutive clean phases, in the
second phase of the pair our algorithm only calls
the covering oracle $\DCovOracle$
with the proper ordinal $\lambda$ (i.e. the
maximum $\lambda$ such that $S_{\lambda}$ contains
an optimal strategy) and with a
set of balls $B_0,B_1,\ldots,B_n$
that covers $U$.
Also, note that an active strategy $v$
during a run of $\PhaseAlg(T,d,\lambda)$
never has a confidence radius $r_t(v)$ less
than $\delta = T^{-1/(d+2)}$, so the strategies
activated by the covering oracle form a
$\delta$-net in the space $S_{\lambda} \setminus U$.
By Definition~\ref{def:fatness-transfinite}(c),
a $\delta$-net in
$S_{\lambda} \setminus U$ contains fewer
than $O(\delta^{-d})$ points.  Hence for
sufficiently large $T$ the
``quota'' of $\tfrac12\; T^{d/(d+2)}$ active
strategies is never reached, which implies
that every point of $S_{\lambda}$ --- including
$u^*$ --- is covered throughout the phase.
The upper bound on the regret of
$\PhaseAlg(T,d,\lambda)$ concludes as
in the proof of Theorem~\ref{thm:zooming-dim}.
\end{proofof} 

\OMIT{
\begin{proof}[ Sketch of Theorem~\ref{thm:PMO}]
The full proof is in Appendix~\ref{apdx:pmo}.  Here we only explain
how the analysis of $\A_d$ differs from the analysis
of the algorithm for finite
$d$-fat decompositions sketched in previous sections.
We look at the deepest $S_{\lambda}$ containing
an optimal strategy $x^*$.  The analysis hinges on showing that
in any two consecutive clean phases, the first one ends by
setting $\lambda(i) = \lambda$ and the second one keeps
$x^* \in S_{\lambda}$ covered throughout.\footnote{In fact,
the peculiar way in which we define $\eps(i)$ and the confidence radius in our algorithm,
and also the collection of balls in the end of each phase,
is chosen precisely in order to produce $\lambda(i)=\lambda$.}
A key ingredient
in both steps is the ``entombment'' of $S_{\lambda+1}$:
there is an open neighborhood $U \supseteq S_{\lambda+1}$
such that $U$ remains covered forever after some finite
time $T_0$.  The fact that the first clean phase ends
by setting the correct $\lambda(i)$ is established using
the entombment of $S_{\lambda+1}$, along with
a pigeonhole principle argument which proves that there exist
frequently-played strategies (hence small balls exist)
and then argues using the modified definition of index
that one of these balls must contain $x^*$.
\end{proof}
}
\OMIT{
Since we have modified the definition of index, we must
prove a variant of Claim~\ref{lm:fatness-confRad} which
asserts that in a clean run of $\PhaseAlg$, if
$u,v$ are active at time $t$ then $\Delta(v) - \Delta(u) \leq
5 r_t(v).$  To prove it, let $s$ be the latest round in
$\{1,2,\ldots,t\}$ when $v$ was played.  We have
$r_t(v) = r_s(v)$, and $\Delta(v) - \Delta(u) =
\mu(u) - \mu(v)$, so it remains to prove that
\begin{equation} \label{eqn:idx5.1}
\mu(u) - \mu(v) \leq 5 r_s(v).
\end{equation}
From the fact that $v$ was played instead of $u$
at time $s$, together with the fact that both
strategies are clean, we have
\begin{align}
\label{eqn:idx5.2}
\mu_s(u) + 4 r_s(u) & \leq \mu_s(v) + 4 r_s(v) \\
\label{eqn:idx5.3}
\mu(u) - \mu_s(u) & \leq  r_s(u) \\
\label{eqn:idx5.4}
\mu_s(v) - \mu(v) & \leq  r_s(v)
\end{align}
and (\ref{eqn:idx5.1}) follows by adding
(\ref{eqn:idx5.2})-(\ref{eqn:idx5.4}) and
using the fact that $r_s(u) > 0.$

Let $\lambda$ be the maximum ordinal such that
$S_{\lambda}$ contains an optimal strategy $u^*$; such an
ordinal exists by Definition~\ref{def:fatness-transfinite}(b).
We will prove that for sufficiently
large $i$, if the $i$-th phase is clean, then
$\lambda(i) = \lambda$.
The set $S_{\lambda+1}$ is compact, and the function $\mu$ is
continuous, so it assumes a maximum value on $S_{\lambda+1}$
which is, by construction, strictly less than $\mu^*$.
Choose $\ve>0$ such that $\Delta(w) > 5 \ve$ for all
$w \in S_{\lambda+1}$, and choose $T_0 = 2^{i_0}$ such
that $\ve(i_0) \leq \ve.$  We shall prove that for all
$T = 2^i \geq T_0$ and all ordinals $\nu$, a clean run
of $\PhaseAlg(T,d,\nu)$ results in setting $\lambda(i+1)=\lambda$.
First, let $v^* \in \mathcal{N}$ be such that
$L(u^*,v^*) \leq \ve(i).$  If $v$ is active and
$r_T(v) < \ve(i)/2$ then Claim~\ref{lm:fatness-confRad}
implies that $\Delta(v) - \Delta(v^*) \leq \tfrac52\, \ve(i)$
hence $\Delta(v) \leq \tfrac72\, \ve(i)$.  As $\Delta(w) > 5 \ve
\geq 5 \ve(i)$ for all $w \in S_{\lambda+1}$, it follows
that the closure of $B(v,\ve(i))$ does not intersect
$S_{\lambda+1}.$  This guarantees that
$\DpthOracle(B_0,B_1,\ldots,B_m)$ returns
an ordinal less than or equal to $\lambda.$
Next we must prove that this ordinal is greater
than or equal to $\lambda.$
Note that the
total number of strategies activated by $\PhaseAlg(T,d,\nu)$
is bounded above by $T^{d/(d+2)} \log(T)$.  Let
$A_T$ denote the set of strategies active at time $T$ and let
\begin{align*}
v^0 &= \arg \max_{v \in A_T} n_T(v).
\end{align*}
By the pigeonhole principle, we have
$n_T(v^0) \geq T^{2/(d+2)} / \log(T)$
hence $r_T(v^0) \leq T^{-1/(d+2)} \log(T).$
If $t$ denotes the last time at which
$v^0$ was played, then we have
\[
I_t(v^0) = \mu_t(v^0) + 4 r_t(v^0)
\leq \mu^* + 5 r_t(v^0) \leq \mu^* + 5 T^{-1/(d+2)} \log(T)
< \mu^* + \ve(i)/2,
\]
provided that the phase is clean and that
$T \geq T_0.$  Since $v^0$ had maximum index
at time $t$, we deduce that $I_t(v^*) < \mu^* + \ve(i)/2$
as well.  As $L(u^*,v^*) \leq \ve(i)$ we have
$\mu_t(v^*) \geq \mu^* - \ve(i) - r_t(v^*)$ provided
the phase is clean.  To finish the proof we observe
that
\[
\mu^* + \ve(i)/2 > I_t(v^*) \geq \mu^* - \ve(i) + 3 r_t(v^*)
\]
which implies $r_t(v^*) < \ve(i)/2$.  Since the confidence
radius does not increase over time, we have $r_T(v^*) < \ve(i)/2$
so $B(v^*,ve(i))$ is one of the balls $B_0,B_1,\ldots,B_m.$
Since $u^*$ is contained in the closure of this ball, we
may conclude that $\DpthOracle(B_0,B_1,\ldots,B_m)$ returns
the ordinal $\lambda$ as desired.

Let $U = B(S_{\lambda+1}, \ve(i)/2)$.  As in
Claim~\ref{lm:fatness-entombment} it holds
that in any clean phase, $U$ is covered throughout
the phase by balls centered at points of $\mathcal{N}$.
Hence for any pair of consecutive clean phases, in the
second phase of the pair our algorithm only calls
the covering oracle $\DCovOracle$
with the proper ordinal $\lambda$ (i.e. the
maximum $\lambda$ such that $S_{\lambda}$ contains
an optimal strategy) and with a
set of balls $B_0,B_1,\ldots,B_n$
that covers $U.$
Also, note that an active strategy $v$
during a run of $\PhaseAlg(T,d,\lambda)$
never has a confidence radius $r_t(v)$ less
than $\delta = T^{-1/(d+2)}$, so the strategies
activated by the covering oracle form a
$\delta$-net in $S_{\lambda} \setminus U$.
By Definition~\ref{def:fatness-transfinite}(c),
a $\delta$-net in
$S_{\lambda} \setminus U$ contains fewer
than $O(\delta^{-d})$ points.  Hence the
``quota'' of $\tfrac12\; T^{d/(d+2)}$ active
strategies is never reached, which implies
that every point of $S_{\lambda}$ --- including
$u^*$ --- is covered throughout the phase.
The upper bound on the regret of
$\PhaseAlg(T,d,\lambda)$ concludes as
in the proof of Theorem~\ref{thm:zooming-dim}.
}


\section{Zooming algorithm: extensions and examples}
\label{sec:gen-confRad}

We extend the analysis in Section~\ref{sec:adaptive-exploration} in several directions, and follow up with examples.

\begin{itemize}
\item In Section~\ref{sec:conf-rad} we note that our analysis works under a more abstract notion of the confidence radius: essentially, it can be any function of the history of playing a given strategy such that Claim~\ref{cl:conf-rad} holds. This observation leads to sharper results if the reward from playing each strategy $u$ is $\mu(u)$ plus an independent \emph{noise} of a known and ``benign" shape; we provide several concrete examples.

\item In Section~\ref{sec:maxReward1} we provide an improved version of the confidence radius such that the zooming algorithm satisfies the guarantee in Theorem~\ref{thm:zooming-dim} {\bf and} achieves a better regret exponent $\tfrac{d}{d+1}$ if the maximal reward is exactly 1. The analysis builds on a novel Chernoff-style bound which, to the best of our knowledge, has not appeared in the literature.

\item In Section~\ref{sec:targetMAB} we consider the an example which show-cases both the notion of the zooming dimension and the improved algorithm from Section~\ref{sec:maxReward1}. It is the \emph{target MAB} problem, a version of the Lipschitz MAB problem in which the expected reward of a given strategy is a equal to its distance to some (unknown) \emph{target set} $S$. We show that the zooming algorithm performs much better in this setting; in particular, if the metric is doubling and $S$ is finite, it achieves \emph{poly-logarithmic} regret.

\item In Section~\ref{sec:noisy-NN-app} we relax some of the assumptions in the Lipschitz MAB problem: we do not require the similarity function $L$ to satisfy the triangle inequality, and we need the Lipschitz condition~\refeq{eq:Lipschitz-condition} to hold only if one of the two strategies is optimal. We use this extension to analyze a generalization of the target MAB problem in which $\mu(u) = f(L(u,S))$ for some known function $f$.

\item Finally, in Section~\ref{sec:BerryEsseen} we extend the analysis in Section~\ref{sec:adaptive-exploration} from reward distributions with bounded support\footnote{In Section~\ref{sec:conf-rad} we also consider \emph{stochastically bounded} distributions such as Gaussians.} to arbitrary reward distributions with a finite absolute third moment. Our analysis relies on the extension of Azuma inequality known as the \emph{non-uniform Berry-Esseen theorem}~\cite{BerryEsseen-survey05}.
\end{itemize}

Let us recap some conventions we'll be using throughout this section. The zooming algorithm proceeds in phases $i=1,2,3,\ldots$ of $2^i$ rounds each. Within a given phase, for each strategy $v\in X$ and time $t$, $n_t(v)$ is the number of times $v$ has been played before time $t$, and  $\mu_t(v)$ is the corresponding average reward. Also, we denote $\Delta(v) = \mu^* - \mu(v)$, where
	$\mu^* = \sup_{v\in X} \mu(v)$
is the maximal reward.

\subsection{Abstract confidence radius and noisy rewards}
\label{sec:conf-rad}

In Section~\ref{sec:conf-rad} the confidence radius of a given strategy was defined by~\refeq{eq:confidence-radius}. Here we generalize this definition to any function of the history of playing this strategy that satisfies certain properties.

\begin{definition}\label{def:conf-rad}
Consider a single phase $i_\Phase$ of the algorithm. For each strategy $v$ and any time $t$ within this phase, let
	$\hat{r}_t(v)$ and $\hat{\mu}_t(v)$
be non-negative functions of $i_\Phase$, $t$, and the history of playing $v$ up to round $t$. Call $\hat{r}_t(v)$ a \emph{confidence radius} with respect to $\hat{\mu}_t(v)$ if
\begin{OneLiners}
\item[(i)] $|\hat{\mu}_t(v) - \mu(t)| \leq \hat{r}_t(v)$ with probability at least $1-8^{-i_\Phase}$.

\item[(ii)] $ \tfrac34\, \hat{r}_t(v) \leq \hat{r}_{t+1}(v)
	\leq \hat{r}_t(v) $.
\end{OneLiners}
The confidence radius is \emph{$(\beta,C)$-good} if
	$n_t(v) \leq (C\, i_\Phase)\, \Delta^{-\beta}(v)$
whenever $\Delta(v) \leq 4 \hat{r}_t(v)$.
\end{definition}

\begin{note}{Remark.}
Property (i) says that Claim~\ref{cl:conf-rad} holds for the appropriately redefined \emph{clean phase}. Property (ii) is a ``smoothness'' condition: $\hat{r}_t(v)$ does not increase with time, and does not decrease too fast. It is needed for the last line of the proof of Lemma~\ref{lm:bound-active}.
\end{note}

\noindent Given such confidence radius, we can carry out the proof of Theorem~\ref{thm:zooming-dim} with very minor modifications.

\begin{theorem}\label{thm:zooming-dim-generalized}
Consider an instance of the \standardMAB\ for which there exists a $(\beta, c_0)$-good confidence radius, $\beta\geq 0$. Let \A\ be an instance of Algorithm~\ref{alg:nice-1d} defined with respect to this confidence radius. Suppose the problem instance has $c$-zooming dimension $d$. Then:
\begin{itemize}
\item[(a)] If $d+\beta>1$ then
	$R_\A(t) \leq   a(t)\; t^{1-1/(d+\beta)} $
for all $t$, where
	$a(t) =  O(c\,c_0\,\log^2 t)^{1/(d+\beta)}$.
\item[(b)] If $d+\beta \leq 1$ then
	$R_\A(t) \leq   O(c\,c_0\,\log^2 t) $.
\end{itemize}
\end{theorem}

\begin{note}{Remark.}
A new feature of this theorem (as compared to Theorem~\ref{thm:zooming-dim}) is the \emph{poly-logarithmic} bound on regret in part (b). For better intuition on this, note that the exponent in part (a) becomes negative if $d+\beta<1$. Since the regret bound should not be \emph{decreasing} in $t$, one would expect this term to vanish from the ``correct'' bound. Indeed, it is easy to check that the computation in the proof of Claim~\ref{cl:clean-phase} results in part (b).
\end{note}

A natural application of Theorem~\ref{thm:zooming-dim-generalized} if a setting in which the reward from playing each strategy $u$ is $\mu(u)$ plus an independent \emph{noise} of known shape.

\newcommand{\PP}{\ensuremath{\mathcal{P}}}

\begin{definition}
The \emph{Noisy Lipschitz MAB problem} is a standard Lipschitz MAB problem such that every time any strategy $u$ is played, the reward is $\mu(u)$ plus an independent random sample from some fixed distribution $\PP$ (called the \emph{noise distribution}) which is revealed to the algorithm.
\end{definition}

We present several examples in which we take advantage of a ``benign" shape of \PP. Interestingly, in these examples the payoff distributions are not restricted to have bounded support.\footnote{Recall that throughout the paper the payoff distribution of each strategy $x$ has support
	$\mathcal{S}(x)\subset [0, 1]$.
In this subsection, by a slight abuse of notation, we do not make this assumption.} Technically the results are simple corollaries of Theorem~\ref{thm:zooming-dim-generalized}.

We start with perhaps the most natural example when the noise distribution is normal.

\begin{corollary}\label{cor:noisy-normal}
Consider the Noisy Lipschitz MAB problem with normal noise distribution
	$\PP = \mathcal{N}(0,\sigma^2)$.
Then there exists an algorithm \A\ which enjous guarantee~\refeq{eq:thm-zooming-dim} with the right-hand side multiplied by $\sigma$.
\end{corollary}

\begin{proof}
Define the confidence radius as~\refeq{eq:confidence-radius} with the right-hand side multiplied by $\sigma$. It is easy to see that this is a $(2, O(\sigma))$-good confidence radius. The result follows from Theorem~\ref{thm:zooming-dim-generalized}(a).
\end{proof}

\begin{note}{Remark.}
In fact, Corollary~\ref{cor:noisy-normal} can be extended to noise distributions of a somewhat more general form: let us say that a random variable $X$ is \emph{stochastically $(\rho, \sigma)$-bounded} if its moment-generating function satisfies
\begin{align}\label{eq:stoc-bdd}
 E[ e^{r (X-E[x])} ] \leq e^{r^2 \sigma^2/2}
    \text{~~for all $r\in [-\rho, \rho]$}.
\end{align}
\noindent Note that a normal distribution $\mathcal{N}(0,\sigma^2)$ is $(\infty, \sigma)$-bounded, and any distribution with support $[-\sigma, \sigma]$ is $(1,\sigma)$-bounded. The meaning of~\refeq{eq:stoc-bdd} is that it is precisely the condition needed to establish an Azuma-type inequality: if $S$ is the sum of $n$ independent stochastically $(\rho, \sigma)$-bounded random variables with zero mean, then with high probability $S\leq \Tilde{O}(\sigma_i \sqrt{n})$:
\begin{equation}\label{eq:intro-Azuma}
 \Pr \left[ S > \lambda \sigma \sqrt{n}  \right]
    \leq \exp(-\lambda^2/2)\quad
    \text{for any $\lambda \leq  \tfrac12\, \rho\,\sigma\sqrt{n} $.}
\end{equation}
The derivation and the theorem statement needs to be modified slightly to account for the parameter $\rho$; we omit the details from this version.
\end{note}

Second, we consider the \emph{noiseless} case when all probability mass in \PP\ is concentrated at 0. Our result holds more generally, when \PP\ has at least one \emph{point mass}: a point $x\in\R$ such that $\PP(x)>0$.

\begin{corollary}\label{cor:point-mass}
Consider the Noisy Lipschitz MAB problem such that the noise distribution \PP\ has at least one point mass. Then the problem admits a confidence radius which is $(\beta,c)$-good for any given $\beta>0$ and a constant $c = c(\beta,\PP)$. The corresponding low-regret guarantees follow via Theorem~\ref{thm:zooming-dim-generalized}.
\end{corollary}

\sketch{
Let $S = \argmax \PP(x)$ be the set of all points with the largest point mass $p = \max_x \PP(x)$, and let $q = \max_{x:\, \PP(x)<p} \PP(x)$ be the second largest point mass. Then $n = \Theta(\log t)$ samples suffices to ensure that with high probability each node in $S$ will get at least $n(p+q)/2$ hits whereas any other node will get less, which exactly locates all points in $S$. We use confidence radius
	$r_t(v) = \Theta(i_\Phase)(\tfrac34)^{n_t(v)}$.
} 

Third, we consider noise distributions with a "special region" which can be located using a few samples. This may be a more efficient way to estimate $\mu(v)$ than using the standard Chernoff-style tail bounds. Moreover, in our examples $\PP$ may be heavy-tailed, so that Chernoff-style bounds do not hold.

\begin{corollary}\label{thm:zooming-dim-uniform}
Consider the Noisy Lipschitz MAB problem with noise distribution \PP. Suppose $\PP$ has a density $f(x)$ which is symmetric around $0$ and non-increasing for $x>0$. Assume one of the following:
\begin{OneLiners}
\item[$(a)$] $f(x)$ has a sharp peak:
	$ f(x) = \Theta(|x|^{-\alpha})$
for all small enough $|x|$, where $\alpha\in (0,1)$.

\item[$(b)$] $f(x)$ piecewise continuous on $(0, \infty)$ with at least one jump.
\end{OneLiners}
\noindent Then for some constant $c_\PP$ that depends only on \PP\  the problem admits a
	$(\beta, c_\PP)$-good
confidence radius, where $(a)$~~$\beta = 1-\alpha$,~~$(b)$~~$\beta=1$. The corresponding low-regret guarantees follow via Theorem~\ref{thm:zooming-dim-generalized}.
\end{corollary}

\sketch{
For part ($a$), note that for any $x>0$ in a neighborhood of $0$ we have
	$\PP[ (-x, x) ] = \Theta(x^{1-\alpha})$.
Therefore
	$n = \Theta( x^{\alpha-1}\,\log t)$
samples suffices to separate with high probability any length-$x$ sub-interval of $(-x, x)$ from any length-$x$ sub-interval of $(2x, \infty)$. It follows that using $n$ samples we can approximate the mean reward up to $\pm O(x)$. Accordingly,  we set
	$r_t(v) = \Theta(i_\Phase/n_t(v))^{1/(1-\alpha)}$.

For part ($b$), let $x_0$ be the smallest positive point where density $f$ has a jump. Then by continuity there exists some $\eps>0$ such that
	$ \inf_{x\in (x_0-\eps,\, x_0)} f(x) > \sup_{x\in (x_0,\, x_0+\eps)} f(x)$.
Therefore for any $x<\eps$ using
	$n = \Theta(\tfrac{1}{x}\log n)$
samples suffices to separate with high probability any length-$x$ sub-interval of $(0, x_0)$ from any length-$x$ sub-interval of
	$(x_0, \infty)$.
It follows that using $n$ samples we can approximate the mean reward up to $\pm O(x)$. Accordingly,  we set
	$r_t(v) = \Theta(i_\Phase/n_t(v)) $.
} 

\OMIT{ 
For part ($b_3$), suppose
	$ f(x) = \Theta(x^{-\alpha})$
for all $x\geq x_0$. Then for all $x\geq x_0$ we have
	$\PP[ (x_0, x_0+x ] = \Theta(x^{1-\alpha})$.
Therefore for any $x>0$ using
	$n = \Theta_{\alpha}(\tfrac{1}{x}\,\log t) $
samples suffices to separate with high probability any length-$x$ sub-interval of $(0, x_0)$ from any length-$x$ sub-interval of
	$(x_0, \infty)$.
It follows that using $n$ samples we can approximate the mean reward up to $\pm O(x)$. So  we set
	$r_t(v) = \Theta(\log t)\, n_t(v)^{-1} $.
} 

\subsection{What if the maximal expected reward is 1?}
\label{sec:maxReward1}

\newcommand{\Rrel}{r}

We elaborate the algorithm from Section~\ref{sec:adaptive-exploration} so that it satisfies the guarantee~\refeq{eq:thm-zooming-dim} {\bf and} performs much better if the maximal expected reward is $1$.

\begin{definition}
Consider the Lipschitz MAB problem. Call an algorithm \emph{$\beta$-good} if there exists an absolute constant $c_0$ such that for any problem instance of $c$-zooming dimension $d$ it has the properties (ab) in Theorem~\ref{thm:zooming-dim-generalized}. Call a confidence radius  \emph{$\beta$-good} if it is \emph{$(\beta,c_0)$-good} for some absolute constant $c_0$.
\end{definition}

\begin{theorem}\label{thm:zooming-dim-maxReward}
Consider the \standardMAB. There is an algorithm \A\ which is $2$-good in general, and $1$-good when the maximal expected reward is $1$.
\end{theorem}

The key ingredient here is a refined version of the confidence radius which is much sharper than~\refeq{eq:confidence-radius} when the sample average  is close to $1$. For phase $i_\Phase$, we define
\begin{equation}\label{eq:conf-rad-relative}
 \Rrel_t(v) := \frac{\alpha}{1+n_t(v)} +
	\sqrt{ \alpha\; \frac{1-\mu_t(v)}{1+n_t(v)}}\;\;
	\text{for some $\alpha= \Theta(i_\Phase)$}.
\end{equation}

In order to analyze~\refeq{eq:conf-rad-relative} we need to establish the following Chernoff-style bound which, to the best of our knowledge, has not appeared in the literature:

\begin{lemma}\label{lm:my-chernoff}
Consider $n$ i.i.d. random variables $X_1 \ldots X_n$ on $[0,1]$. Let $\mu$ be their mean, and let $X$ be their average. Then for any $\alpha>0$ the following holds:
$$ \Pr\left[\, |X-\mu| < r(\alpha,X) < 3\,r(\alpha,\mu) \, \right] > 1-e^{-\Omega(\alpha)},\;
	\text{where $r(\alpha,x)
		= \tfrac{\alpha}{n} + \sqrt{\tfrac{\alpha x}{n}}$}.
$$
\end{lemma}

\begin{proof}
We will use two well-known Chernoff Bounds which we state below (e.g. see p. 64 of~\cite{MitzUpfal-book05}):
\newcommand{\CB}{{\sc cb}}
\begin{itemize}
\item[(\CB1)] $\Pr[ |X-\mu| > \delta \mu ] < 2\, e^{-\mu n \delta^2/3} $
	for any $\delta\in (0,1)$.
\item[(\CB2)] $\Pr[ X > a ] < 2^{-an} $
	for any $a>6\mu$.
\end{itemize}
First, suppose
	$\mu\geq \tfrac{\alpha}{6n} $. Apply (\CB1) with
	$\delta = \tfrac12 \sqrt{\tfrac{\alpha}{6\mu n}}$.
Thus with probability at least $1- e^{-\Omega(\alpha)}$ we have
	$|X-\mu| < \delta\mu \leq \mu/2$.
Moreover, plugging in the value for $\delta$,
$$ |X-\mu|
	< \tfrac12  \sqrt{\alpha \mu/n}
	\leq \sqrt{\alpha X/n}
	\leq r(\alpha, X) < 1.5\, r(\alpha, \mu).
$$

Now suppose $\mu< \tfrac{\alpha}{6n}$. Then using (\CB2) with $a = \tfrac{\alpha}{n}$, we obtain that with probability at least $1- 2^{-\Omega(\alpha)}$ we have
	$X < \tfrac{\alpha}{n}$,
and therefore
$$ |X-\mu| < \tfrac{\alpha}{n} < r(\alpha, X) <  (1+\sqrt{2})\, \tfrac{\alpha}{n}
	< 3\, r(\alpha, \mu). \qedhere
$$
\end{proof}

\begin{proofof}{Theorem~\ref{thm:zooming-dim-maxReward}}
Let us fix a strategy $v$ and time $t$. Let us use Lemma~\ref{lm:my-chernoff} with $n=n_t(v)$ and $\alpha = \Theta(i_\Phase)$ as in~\refeq{eq:conf-rad-relative}, setting each random variable $X_i$ equal to 1 minus the reward from the $i$-th time strategy $v$ is played in the current phase. Then $\mu = \mu(v)$ and $X = \mu_t(v)$, so the Lemma says that
\begin{align}\label{eq:max1-pf}
 \Pr\left[\,
	|\mu_t(v)-\mu(v)| < r_t(v)
	< 3\left(
			\frac{\alpha}{n_t(v)} + \sqrt{\frac{\alpha\, (1-\mu(v))}{n_t(v)}}\,
 	   \right)
 	\right]
  > 1-2^{\Omega(\alpha)}.
\end{align}

Note that~\refeq{eq:conf-rad-relative} is indeed a confidence radius with respect to $\mu_t(v)$: property (i) in Definition~\ref{def:conf-rad} holds by~\refeq{eq:max1-pf}, and it is easy to check that property (ii) holds, too. It is easy to see that~\refeq{eq:conf-rad-relative} is a $2$-good confidence radius. It remains to show that it is $1$-good when the maximal reward is $1$; this is where we use the upper bound on $r_t(v)$ in~\refeq{eq:max1-pf}. It suffices to prove the following claim:
$$\text{
If the maximal reward is $1$ and
	$\Delta(v) \leq 4\, \Rrel_t(v) $
then
	$n_t(v) \leq O(\log t)\, \Delta(v)^{-1}$.
}$$
Indeed, let $n = n_t(v)$ and $\Delta = \Delta(v)$, and suppose that the maximal reward is $1$ and
	$\Delta(v) \leq 4\, \Rrel_t(v) $.
Then by~\refeq{eq:max1-pf} we have
$\Delta \leq 4\, \Rrel_t(v)
	\leq \tfrac{\alpha}{n} + \sqrt{\alpha\Delta/n} $
for some
	$\alpha = O(\log t)$.
Now there are two cases.
If $ \tfrac{\alpha}{n} < \Delta/2$ then
	$ \sqrt{\alpha \Delta/n} \geq \Delta - \tfrac{\alpha}{n} > \Delta(v)/2$,
which implies the desired inequality. Else we simply have
	$n \leq O(\alpha/\Delta)$.
Claim proved.
\end{proofof}

\subsection{Example: expected reward $=$ distance to the target}
\label{sec:targetMAB}

We consider a version of the Lipschitz MAB problem where the expected reward of a given strategy is equal to its distance to some \emph{target set} which is not revealed to the algorithm.

\OMIT{ 
\begin{definition}\label{def:metric-MAB}
A \emph{\quasimetric} on a set $X$ is a symmetric function
	$L: X\times X \rightarrow [0, \infty)$
such that $L(x,x)=0$ for all $x$.
A \emph{metric MAB problem} on a \quasimetric\ $(L,X)$ is an i.i.d. MAB problem on a strategy set $X$ such that the algorithm is given a \emph{covering oracle} for $(L,X)$.
\end{definition}
 and  $f: [0,1]\rightarrow [0,1]$ is a non-decreasing \emph{shape function} known to the algorithm.
} 

\begin{definition}\label{def:noisyNN}
The Target MAB problem on a metric space $(L,X)$ with a \emph{target set} $S\subset X$ is the standard Lipschitz MAB problem on $(L,X)$  with payoff function
	$\mu(u) = 1-L(u,S)$.
\end{definition}

\begin{note}{Remark.}
It is a well-known fact that
	$ L(u,v) \geq L(u,S) - L(v,S)$
for any $u,v\in X$ and any set $S\subset X$. Therefore the payoff function $\mu$ in Definition~\ref{def:noisyNN} is Lipschitz on $(L,X)$.
\end{note}

Note that in the Target MAB problem the maximal reward is 1, so we can take advantage of the zooming algorithm \A\ from Theorem~\ref{thm:zooming-dim-maxReward}. Recall that
	$R_\A(t) \leq \tilde{O}(c\,t^{1-1/(1+d)})$
where $d$ is the $c$-zooming dimension. In this example zooming dimension is about covering $B(S,r)$ with sets of diameter $\Theta(r)$: it is the smallest $d$ such that for each $r>0$ the ball $B(S,r)$ can be covered with $c\,r^{-d}$ sets of diameter $\leq r/8$.

Let us refine this bound for metric spaces of finite doubling dimension. In particular, we show that for a finite target set the zooming algorithm from Theorem~\ref{thm:zooming-dim-maxReward} achieves \emph{poly-logarithmic} regret.
\begin{theorem}\label{thm:targetMAB}
Consider the Target MAB problem on a metric space of finite doubling dimension $d^*$. Let \A\ be the zooming algorithm from Theorem~\ref{thm:zooming-dim-maxReward}.  Then
\begin{align}\label{eq:targetMAB}
 R_{\A}(t) \leq (c\, 2^{O(d^*)}\log^2 t)\; \; t^{1-1/(1+d)} \;\;\;
	\text{for all $t$},
\end{align}
where $d$ is the $c$-covering dimension of the target set $S$.
\end{theorem}

\begin{proof}
By Theorem~\ref{thm:zooming-dim-maxReward} it suffices to prove that the $K$-zooming dimension of the pair $(L,\mu)$ is at most $d$, for some $K = c\, 2^{O(d^*)}$. In other words, it suffices to cover the set
	$S_\delta = \{ u\in Y: \Delta(u) \leq \delta \}$
with $K\, \delta^{-d}$ sets of diameter $\leq \delta/16$, for any given $\delta>0$.

Fix $\delta>0$ and note that
	$ \Delta(u) = L(u,S)$.
Note that set $S$ can be covered with $c\,\delta^{-d}$ sets
	$\{\, C_i \,\}_i$
of diameter $\leq \delta$. It follows that the set $S_\delta$ can be covered with $r^{-d}$ sets
	$\{\, B(C_i, r) \,\}_i$
of diameter $\leq 3r $. Moreover, each set $B(C_i, r)$ can be covered with $2^{O(d^*)}$ of sets of diameter $\leq \delta/16$.
\end{proof}

\begin{note}{Remarks.}
This theorem is useful when $d<d^*$, i.e. when the target set is a low-dimensional subset of the metric space. Recall that the zooming algorithm is self-tuning: it does not need to know $d^*$ and $d$, and in fact it does not even need to know that it is presented with an instance of the Target MAB problem!
\end{note}

We note in passing that it is very easy to extend Theorem~\ref{thm:targetMAB} to a setting in which the strategy set $Y$ is a proper subset of the metric space $(L,X)$ and does not contain the target set $S$. If $L(Y,S)=0$ then the guarantee~\refeq{eq:targetMAB} holds as is. If $L(Y,S)>0$ then the following guarantee holds:
$$ R_{\A}(t) \leq (c\, 2^{O(d^*)}\log^2 t)\; \; t^{1-1/(2+d)} \;\;\;
	\text{for all $t$},
$$
where $d$ is the $c$-covering dimension of the set $B(S,r)$, $r = L(Y,S)$.


\subsection{The Lipschitz MAB problem under relaxed assumptions}
\label{sec:noisy-NN-app}

The analysis in Section~\ref{sec:adaptive-exploration} does not require all the assumptions in the Lipschitz MAB problem. In fact, it never uses the triangle inequality, and applies the Lipschitz condition~\refeq{eq:Lipschitz-condition} only if (essentially)  one of the two strategies in~\refeq{eq:Lipschitz-condition} is optimal. Let us formulate our results under the properly relaxed assumptions. In what follows, the \emph{zooming algorithm} will refer to the algorithm in Theorem~\ref{thm:zooming-dim-maxReward}.

\begin{theorem}\label{thm:metricMAB}
Consider a version of the Lipschitz MAB problem on $(L,X)$ in which the similarity metric is not required to satisfy triangle inequality,\footnote{Formally, we require $L$ to be a symmetric  function $X\times X\rightarrow [0, \infty]$ such that $L(x,x)=0$ for all $x\in X$. We call such function a \emph{\quasimetric} on $X$.} and the Lipschitz condition~\refeq{eq:Lipschitz-condition} is replaced by
\begin{align}\label{eq:relaxedLipschitz-A}
 (\forall u\in X)\quad
	\Delta(u) \leq L(u,v^*)\quad
	\text{for some $v^* = \argmax_{v\in X} \mu(v)$}
\end{align}
More generally, if such node $v^*$ does not exist, assume
\begin{align}\label{eq:relaxedLipschitz}
	(\forall \eps>0)\quad
	(\exists v^*\in X) \quad
	(\forall u\in X)\quad
	\Delta(u) \leq L(u,v^*) + \eps.
\end{align}
Then the guarantees for the zooming algorithm in Theorem~\ref{thm:zooming-dim-maxReward} still hold.
\end{theorem}

We apply this theorem to a generalization of the Target MAB problem in which
	$\mu(u) = f(L(u,S))$
for some known non-decreasing \emph{shape function}
	$f:[0,1] \rightarrow [0,1]$.
Let us define a \quasimetric\ $L_f$ by
	$L_f(u,v) = f(L(u,v)) - f(0)$.
It is easy to see that $L_f$ satisfies~\refeq{eq:relaxedLipschitz-A}. Indeed, fix any $u^*\in S$. Then
$$ \Delta(u) = f(L(u,S))-f(0)
	\leq f(L(u,u^*)) - f(0)
	= L_f(u,u^*)
	\quad\quad (\forall u\in X).
$$
Thus we can use the zooming algorithm on the \quasimetric\ $L_f$ and enjoy the guarantees in Theorem~\ref{thm:zooming-dim-maxReward}. Below we refine these guarantees for several examples.

Our goal here is to provide clean illustrative statements rather than cover the most general setting to which our refined guarantees apply. Therefore we start with the most concrete example which we formulate as a theorem, and follow up with some extensions which we list without a proof.

\begin{theorem}\label{thm:targetMAB-shape}
Consider the Target MAB problem on a metric space $(L,X)$ of finite doubling dimension $d^*$, with shape function
	$f(x) = x^{1/\alpha}$, $\alpha>0$.
Let \A\ be the zooming algorithm on $(L_f,X)$. Then
\begin{align}\label{eq:targetMAB-shape}
 R_{\A}(t) \leq (c\, 2^{O(d^*)}\log^2 t)\; \; t^{1-1/(1+\alpha d)} \;\;\;
	\text{for all $t$},
\end{align}
where $d$ is the $c$-covering dimension of the target set $S$.
\end{theorem}

\begin{proof}
Consider the pair $(L_f, \mu)$. Since the maximal reward is 1, by Theorem~\ref{thm:zooming-dim-maxReward} it suffices to prove that for some $c^* =c\, 2^{O(d^*)}$ the $c^*$-zooming dimension of this pair is at most $\alpha d$. Specifically, for each $\delta>0$ we need to cover the set
	$S_\delta = \{u\in X:\, \Delta(u)\leq \delta\}$
with
	$c^*\,\delta^{-\alpha d} $
sets of $L_f$-diameter at most $\delta/16$.

Indeed, since $\Delta(u) = L_f(u,S)$, for each
	$u\in S_\delta$ we have $ L(u,S) \leq \delta^\alpha$.
Thus $S_\delta \subset B(S, \delta^\alpha)$. Let $\eps = 16^{-\alpha}$. As in the proof of Theorem~\ref{thm:targetMAB}, we can show that $S_\delta$ can be covered by
	$c\, \eps^{-O(d^*)}\,\delta^{-\alpha d} $
sets of diameter $\eps\, \delta^\alpha$. Each of these sets has $L_f$-diameter at most $f(\eps\, \delta^\alpha)$, which is at most $\delta/16$.
\end{proof}

\begin{note}{Remarks.}
This theorem includes Theorem~\ref{sec:targetMAB} as a special case $f(x) = x$. Like the latter, this theorem is useful when the target set is a low-dimensional subset of the metric space. \end{note}

We consider extensions to more general shape functions and to strategy sets which do not contain $S$:

\begin{itemize}

\item Suppose the shape function $f$ satisfies the following constraints for some constants $\alpha \geq \alpha^* > 0$:
$$  \forall x\in(0,1] \quad
	g(x) \geq x^{1/\alpha}
	\text{ and }
	g(x) \geq 2^{1/\alpha^*} g(\tfrac{x}{2}),
$$
where $g(x) = f(x)-f(0)$. Then for
	$\beta = 1+ 1_{\{ f(0)>0 \}}$
we have
$$
 R_{\A}(t) \leq (c\, 2^{O(\alpha^* d^*/\alpha)}
 		\log^2 t)\; \; t^{1-1/(\beta+ \alpha\,d)}\;\;\;
	\text{for all $t$}.
$$

\item Consider the setting in which the strategy set $Y$ is a proper subset of the metric space $(L,X)$ and does not contain the target set $S$. If $L(u^*,S) = 0$ for some $u^*\in Y$ then the guarantee~\refeq{eq:targetMAB-shape} holds as is. In general, if we restrict the shape function to $f(x) = c+x^{1/\alpha}$, $\alpha\in (0,1]$ then
$$ R_{\A}(t) \leq (c\, 2^{O(d^*)}\log^2 t)\; \; t^{1-1/(2+d)} \;\;\;
	\text{for all $t$},
$$
where $d$ is the $c$-covering dimension of the set $S^* = B(S,r)$, $r = L(Y,S)$. Moreover, one can prove similar guarantees with $d^*$ being the doubling dimension of an open neighborhood of $S^*$, rather than that of the entire metric space
\end{itemize}

\subsection{Heavy-tailed reward distributions}
\label{sec:BerryEsseen}

Consider the Lipschitz MAB problem and let $X_n(v)$ be the reward from the $n$-th trial of strategy $v$. The current problem formlulation restricts $X_n(v)$ to support $[0,1]$. In this section we remove this restriction. In fact, it suffices to assume that $X_n(v)$ is an independent random variable with mean $\mu(v)\in [0,1]$ and a uniformly bounded bounded absolute third moment. Note that different trials of the same strategy $\{X_n(v): n\in \N\}$ do not need  to be identically distributed.

\begin{theorem}\label{thm:zooming-dim-BerryEsseen}
Consider the \standardMAB. Let $X_n(v)$ be the reward from the $n$-th trial of strategy $v$. Assume that each $X_n(v)$ is an independent random variable with mean $\mu(v)\in [0,1]$ and furthermore that
	$ E\left[\,|X_n(v)|^3\,\right] < \rho$
for some constant $\rho$. Then there is an algorithm \A\ such that if for some $c$ the problem instance has $c$-zooming dimension $d$ then
\begin{equation}\label{eq:thm-zooming-dim-2}
	R_\A(t) \leq a(t)\; t^{1-1/(3d+6)}\;\;
	\text{for all $t$, where $a(t) = O(c \rho \log t)^{1/(3d+6)}$.}
\end{equation}
\end{theorem}

The proof relies on the non-uniform Berry-Esseen theorem (e.g. see~\cite{BerryEsseen-survey05} for a nice survey) which we use to obtain a tail inequality similar to Claim~\ref{cl:conf-rad}: for any $\alpha>0$
\begin{equation}\label{eq:conf-radius-BerryEsseen}
 \Pr[ |\mu_t(v)-\mu(v)| > r_t(v) ] < O(t^{-3\alpha}),
	\text{where $r_t(v) = \Theta(t^\alpha)/ \sqrt{n_t(v)}$}.
\end{equation}
However, this inequality gives much higher failure probability than Claim~\ref{cl:conf-rad}; in particular, we cannot take a union bound over all active strategies. Accordingly, we need a more refined version of Theorem~\ref{thm:zooming-dim-generalized} which is parameterized by the failure probability in~\refeq{eq:conf-radius-BerryEsseen}. In the analysis, instead of the failure events when the phase is not clean (see Definition~\ref{def:clean-phase}) we need to consider the \emph{$\rho$-failure events} when the tail bound from~\refeq{eq:conf-radius-BerryEsseen} is violated by some strategy $v$ such that $\Delta(v)>\rho$. Then using the technique from Section~\ref{sec:adaptive-exploration} we can upper-bound $R_\A(T)$ in terms of $T$, $d$, $\rho$ and $\alpha$ and choose the optimal values for $\rho$ and $\alpha$.


\newpage
\begin{small}
\bibliographystyle{plain}
\bibliography{bib-abbrv,bib-bandits,bib-rwalks}
\end{small}

\end{document}